\documentclass[reqno,11pt]{amsart}
\usepackage{amsmath, latexsym, amsfonts, amssymb, amsthm, amscd,fig4tex,caption}

\usepackage{graphics}
\usepackage{graphicx}
\usepackage{geometry}
\usepackage{epsfig}
\usepackage{amssymb}

\newif\ifpdf
\ifx\pdfoutput\undefined
\pdffalse
\else
\pdfoutput=1
\pdftrue
\fi
\ifpdf
\usepackage{graphicx}
\usepackage{epstopdf}
\else
\usepackage{graphicx}
\fi

\def\un#1{\hbox{{\indic 1}$_{#1}$}}

\numberwithin{equation}{section}

\newtheorem{theorem}{Theorem}[section]
\newtheorem{lemma}[theorem]{Lemma}
\newtheorem{prop}[theorem]{Proposition}

\newcommand{\und}{\underline}

\newcommand{\gep}{\varepsilon}       

\newcommand{\cB}{{\ensuremath{\mathcal B}} }
\newcommand{\cF}{{\ensuremath{\mathcal F}} }

\newcommand{\cH}{{\ensuremath{\mathcal H}} }
\newcommand{\cG}{{\ensuremath{\mathcal G}} }
\newcommand{\cI}{{\ensuremath{\mathcal I}} }

\newcommand{\ud}{\mathrm{d}\ }

\newcommand{\bP}{{\ensuremath{\mathbf P}} }

\newcommand{\bq}{{\ensuremath{\mathbf q}} }
\newcommand{\bp}{{\ensuremath{\mathbf p}} }
\newcommand{\bn}{{\ensuremath{\mathbf n}} }

\newfont{\indic}{bbmss12}

\def\bo#1{\hbox{{\indic 1}$_{#1}$}}

\newcommand{\hf}{{\frac{1}{2}}}

\newcommand{\E}{{\ensuremath{\mathbb E}} }

\newcommand{\bbN}{{\ensuremath{\mathbb N}} }
\newcommand{\N}{{\ensuremath{\mathbb N}} }

\newcommand{\bbP}{{\ensuremath{\mathbb P}} }

\newcommand{\R}{{\ensuremath{\mathbb R}} }

\newcommand{\myeqnarray}[1]{
  \begingroup
  \jot=#1pt
  \arraycolsep=2pt
  \begin{eqnarray}}

\newcommand{\eeqnarray}{\end{eqnarray}\endgroup}
\newcommand{\ba}{\begin{array}{cc}}
\newcommand{\ea}{\end{array}}

\renewcommand\H{\operatorname{H}}

\title{Large deviations of the current in stochastic collisional dynamics.}

\author[R.\ Lefevere]{Rapha\"el Lefevere}
 \address{Laboratoire de Probabilit\'es
  et Mod\`eles Al\'eatoires (CNRS UMR 7599), Universit\'e Paris 7
  -- Denis Diderot, UFR Math\'ematiques, Case 7012, B\^atiment
  Chevaleret, 75205 Paris Cedex 13, France}
\email{lefevere\@@math.jussieu.fr}

\author[M.\ Mariani]{Mauro Mariani}
\address{Laboratoire d'Analyse,
Topologie, Probabilit\'es (CNRS UMR 6632), Universit\'e Paul
C\'e\-zanne Aix-Marseille 3, Fa\-cult\'e des Sciences et Techniques
Saint-J\'er\^ome, Avenue Escadrille Normandie-Niemen 13397 Marseille
Cedex 20, France}
\email{mariani@cmi.univ-mrs.fr}

 \author[L. Zambotti]{Lorenzo Zambotti}
 \address{Laboratoire de Probabilit{\'e}s
   et Mod\`eles Al\'eatoires (CNRS UMR. 7599)
 Universit\'e Paris 6 -- Pierre et Marie Curie,
U.F.R. Math\'ematiques, Case 188, 4 place
   Jussieu, 75252 Paris cedex 05, France }
 \email{lorenzo.zambotti\@@upmc.fr}


\date{\today}

\begin{document}
\begin{abstract}
We consider a class of deterministic local collisional dynamics, showing how
to approximate them by means of stochastic models and then studying the fluctuations of the current of energy.  We show first that the variance of the time-integrated current is finite and related to the conductivity by the Green-Kubo relation.  Next we show that the law of the empirical average current satisfies a large deviations principle and compute explicitly the rate functional in a suitable scaling limit.  We observe that this functional is not strictly convex.
\end{abstract}
\maketitle

\section{Introduction}
The paper \cite{LefevereZambotti1} introduced a class of stochastic dynamics aimed at modeling Hamiltonian dynamics describing aerogels.  Those stochastic dynamics describe also pretty accurately the behaviour of deterministic models made of tracer particles and fixed scatterers  \cite{ EckmannYoung, EckmannYoung2,Larralde,LinYoung,Mejia}. In non-equilibrium statistical mechanics, it is now a well established fact that the large deviations functional of the currents of the relevant conserved quantities play a role similar to the thermodynamic potentials in equilibrium, see \cite{jona0,BD2} for general overviews of the subject.  In this paper, our goal is to go one step further than in \cite{LefevereZambotti1} and study the fluctuations and large deviations properties of the energy current of one of the models considered there: the {\it confined} tracers.  A confined tracer is a particle that moves freely in an interval  and is  reflected at the boundaries of the interval with a random speed  $v$ distributed according to the distribution:
$$
\varphi_{\beta}(v)=\beta\, v \,e^{-\beta \frac{v^2}{2}}, \qquad v>0.
$$
This random reflection models the action of a ``hot" wall and  $\beta$ is the inverse of the temperature of the wall where the collision has taken place.  We first explain the formal connection between the deterministic and the stochastic dynamics. In particular we derive the expression of the infinitesimal generator of the stochastic dynamics, that is an example of piecewise deterministic dynamics.  Next, the Green-Kubo relation for the stochastic dynamics is rigorously derived.

Finally, we obtain our main result: we show that the law of the empirical average current satisfies a large deviations principle and compute its scaling limit.  
A striking feature of the rate function is that it contains linear (affine) parts.  Namely, we obtain the following limit rate function:
\begin{equation}
\cG(j,\tau,T):=\left\{\begin{array}{ll}
\frac{(j-\kappa\tau)^2}{4\kappa T^2} \ \ {\rm if} \ \ j\tau>\kappa\tau^2  \\
0  \ \ {\rm if} \ \ j\tau\in[0,\kappa\tau^2] 
\\ 
\frac{-j\tau}{2 T^2} \ \ {\rm if} \ \ j\tau\in[-\kappa\tau^2,0] 
\\
\\
\frac{j^2+\kappa^2\tau^2}{4\kappa T^2} \ \ {\rm if} \ \ j\tau<-\kappa\tau^2,
\end{array}
\right.
\end{equation}
where $\tau$ is the rescaled temperature gradient applied at the boundaries of the interval and $T$ the arithmetic mean of the temperatures. The graph of $\cG$ as a function $j$ is given in Figure 1 below. In an accompanying paper \cite{LMZ2}, we discuss the implications of this result for the macroscopic fluctuation theory \cite{jona1,jona2,jona0,BD1,BD2,BD3} of deterministic systems modeling aerogels.
A similar large deviations functional has been found in the context of random walks in random environments \cite{Comets,Greven}.  

\begin{figure}[thb]
\includegraphics[width =.50\textwidth]{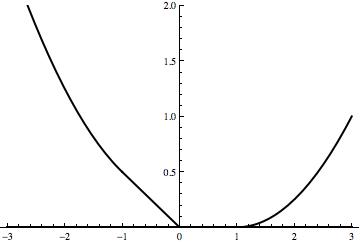} 
\caption*{ Figure 1: Plot of $\cG$ as a function of $j$ for $\kappa\tau=\kappa T^2=1$}
\label{flatld}
\end{figure}

In order to obtain the large deviations principle and the rate function, the main difficulty to overcome is the lack of strict convexity of the candidate rate functional $\cG(j,\tau,T)$, which is obtained as the Legendre transform of a cumulant generating function. In this situation the G\"artner-Ellis theorem does not yield the full large deviations principle and we have to use a more detailed analysis, in particular in order to obtain the lower bound.

  \section{From deterministic to stochastic dynamics.}
  
\medskip
Consider $N$ particles  of unit mass with positions and momenta
$(\underline{\mathbf{q}}, \underline{\mathbf{p}}) \equiv
\big\{(\mathbf{q}_i, \mathbf{p}_i)\big\}_{1\leq i\leq N}$, with
$\mathbf{q}_i, \mathbf{p}_i \in \mathbb{R}^d$.  The positions are measured with respect to $N$ fixed centers located on a 1D lattice. The Hamiltonian $H$ takes
the form,
\begin{equation}
H(\underline{\mathbf{p}}, \underline{\mathbf{q}})
= \sum_{i=1}^N \left[\frac{\mathbf{p}_i^2}{2} +V(\mathbf{q}_i)+
U(\mathbf{q}_{i}-\mathbf{q}_{i+1}) \right],
\label{Hamilton}
\end{equation}
where the interaction potential $U$ is equal to zero inside a region $\Omega_U\subset\mathbb{R}^d$ with
smooth boundary $\Lambda$ of dimension $d-1$, and equal to infinity outside.
Likewise, the pinning potential $V$ is assumed to be zero inside a
bounded region $\Omega_V$ and infinity outside, implying that the motion of
a single particle remains confined for all times. The regions $\Omega_U$ and
$\Omega_V$ being specified, the dynamics is equivalent to a billiard in
high dimension.  A typical example of the dynamics we wish to consider is given by the figure below.  The circles move freely within their square cells and collide with each other when they both get sufficiently close to the hole located in the wall separating two adjacent cells.
\vspace{5mm}

\figinit{0.35mm}
\figpt 1:(-210, 60)
\figpt 2:(-210, 0)
\figpt 3:(210, 0)
\figpt 4:(210,60)

\figpt 5:(-150,60)
\figpt 6:(-90,60)
\figpt 7:(-30,60)
\figpt 8:(30,60)
\figpt 9:(90,60)
\figpt 10:(150,60)

\figpt 11:(-150,35)
\figpt 12:(-90,35)
\figpt 13:(-30,35)
\figpt 14:(30,35)
\figpt 15:(90,35)
\figpt 16:(150,35)

\figpt 17:(-150,25)
\figpt 18:(-90,25)
\figpt 19:(-30,25)
\figpt 20:(30,25)
\figpt 21:(90,25)
\figpt 22:(150,25)

\figpt 23:(-150,0)
\figpt 24:(-90,0)
\figpt 25:(-30,0)
\figpt 26:(30,0)
\figpt 27:(90,0)
\figpt 28:(150,0)

\figpt 29:(-180,35)
\figpt 291:(-165,20)

\figpt 30:(-130,10)
\figpt 301:(-115,35)

\figpt 31:(-50,45)
\figpt 311:(-65,30)

\figpt 32:(-10,20)
\figpt 321:(10,35)

\figpt 33:(57,15)
\figpt 331:(70,35)

\figpt 34:(135,31)
\figpt 341:(120,10)

\figpt 35:(190,41)
\figpt 351:(170,20)

%
\psbeginfig{}
\psset arrowhead(fillmode=yes) \psarrow[29,291]
\psset arrowhead(fillmode=yes) \psarrow[30,301]
\psset arrowhead(fillmode=yes) \psarrow[31,311]
\psset arrowhead(fillmode=yes) \psarrow[32,321]
\psset arrowhead(fillmode=yes) \psarrow[33,331]
\psset arrowhead(fillmode=yes) \psarrow[34,341]
\psset arrowhead(fillmode=yes) \psarrow[35,351]
\psline[1,2,3,4,1]
\psline[5,11]
\psline[6,12]
\psline[7,13]
\psline[8,14]
\psline[9,15]
\psline[10,16]

\psline[17,23]
\psline[18,24]
\psline[19,25]
\psline[20,26]
\psline[21,27]
\psline[22,28]

\pscirc 29(10)
\pscirc 30(10)
\pscirc 31(10)
\pscirc 32(10)
\pscirc 33(10)
\pscirc 34(10)
\pscirc 35(10)

\psendfig
%

%
\figvisu{\figBoxA}{Figure 2: Simplified one-dimensional aerogel dynamics}
{
\figsetmark{$\figBullet$}
}
\centerline{\box\figBoxA}

\vspace{5mm}
Physically, those models describe aerogels, i.e. gels whose liquid components have been removed and replaced by atoms of gases.
In this section, we describe  how to replace the interaction of one particle with its nearest-neigbors by an interaction with local stochastic heat baths.  The model is an approximation based on macroscopic fluctuation theory and a local equilibrium picture:  the action of the whole system on a single atom through its nearest neighbors is replaced by the action of infinite thermal baths on which the particle motion itself has no direct influence on the microscopic time scale \cite{LMZ2}.  One should expect the approximation to hold true whenever numerics show that the systems are close to local equilibrium, i.e. in a weakly interacting regime, see \cite{GG,GL}. This approximation should also apply to some models considered in \cite{ EckmannYoung, EckmannYoung2,Larralde,LinYoung,Mejia} as explained in \cite{LMZ2}.

\subsection{Deterministic dynamics.}
We denote by $\Lambda_U$ and $\Lambda_V$ the boundaries of the regions $\Omega_U$ et $\Omega_V$.  We define  also:
$$
\Omega=\{(\bq_1,\ldots,\bq_N) | \,(\bq_i-\bq_{i+1})\in\Omega_U,\;i=1,\ldots, N-1,\;\bq_i\in\Omega_V, \;i=1,\ldots, N\}.
$$

\noindent{\it Example.}  The simplest example of the type of dynamics we are interested in is given
by a limit of models with smooth interaction potentials. 
In the definition of the Hamiltonian (\ref{Hamilton}),  take $d=1$ and replace the potentials $V$ and $U$ by $V_k$ and $U_k$ where
$V_k(x) = f_k(\frac{x}{b})$ and $U_k(x)=f_k(\frac{x}{a})$, where $f_k(x)=x^{2k}/2k$.
When $k\rightarrow\infty$, the limit interaction potentials $U_\infty$ and $V_\infty$ are  square wells of
infinite heights whose walls are located respectively at $\pm b$ and $\pm
a$. 
\begin{eqnarray}
\label{eq: well}
V_{\infty}(x)=\left\{
\begin{array}{l}
  +\infty\; {\rm if}\; |x|>b\\
  0\; {\rm if }\; |x|\leq b
\end{array}
\right. \qquad U_{\infty}(x)=\left\{
\begin{array}{l}
  +\infty\; {\rm if}\; |x|>a\\
  0\; {\rm if }\; |x|\leq a
\end{array}
\right.
\end{eqnarray}

Each particle on the lattice moves freely on a one-dimensional cell of size
$2 b$, changing directions at the boundaries. The interaction 
between a pair of particles acts when the difference between the positions
of the two particles reaches the value $a$, at which point they exchange
their velocities. 
\vspace{5mm}
\figinit{0.4mm}

\figpt 5:(-160,60)
\figpt 6:(-130,60)
\figpt 7:(-160,0)
\figpt 8:(-130,0)

\figpt 56:(-145,60)
\figpt 78:(-145,0)

\figpt 57:(-160,30)

\figpt 9:(-110,60)
\figpt 10:(-80,60)
\figpt 11:(-110,0)
\figpt 12:(-80,0)

\figpt 13:(-60,60)
\figpt 14:(-30,60)
\figpt 15:(-60,0)
\figpt 16:(-30,0)

\figpt 17:(-10,60)
\figpt 18:(20,60)
\figpt 19:(-10,0)
\figpt 20:(20,0)

\figpt 21:(40,60)
\figpt 22:(70,60)
\figpt 23:(40,0)
\figpt 24:(70,0)

\figpt 25:(90,60)
\figpt 26:(120,60)
\figpt 27:(90,0)
\figpt 28:(120,0)

\figpt 29:(140,60)
\figpt 30:(170,60)
\figpt 31:(140,0)
\figpt 32:(170,0)

\figpt 2930:(155,60)
\figpt 3132:(155,0)

\figpt 41:(-145,45)
\figpt 411:(-145,25)

\figpt 42:(-95,27)
\figpt 421:(-95,47)

\figpt 43:(-45,36)
\figpt 431:(-45,20)

\figpt 44:(5,18)
\figpt 441:(5,35)

\figpt 45:(55,33)
\figpt 451:(55,20)

\figpt 46:(105,45)
\figpt 461:(105,30)

\figpt 47:(155,8)
\figpt 471:(155,28)

\psbeginfig{}

\psset arrowhead(fillmode=yes) \psarrow[41,411]
\psset arrowhead(fillmode=yes) \psarrow[42,421]
\psset arrowhead(fillmode=yes) \psarrow[43,431]
\psset arrowhead(fillmode=yes) \psarrow[44,441]
\psset arrowhead(fillmode=yes) \psarrow[45,451]
\psset arrowhead(fillmode=yes) \psarrow[46,461]
\psset arrowhead(fillmode=yes) \psarrow[47,471]

\psline[5,6]
\psline[7,8]

\psline[9,10]
\psline[11,12]

\psline[13,14]
\psline[15,16]

\psline[17,18]
\psline[19,20]

\psline[21,22]
\psline[23,24]

\psline[25,26]
\psline[27,28]

\psline[29,30]
\psline[31,32]

\psset arrowhead(fillmode=yes) \psarrow[5,7]
\psset arrowhead(fillmode=yes) \psarrow[7,5]


\pscirc 41(5)
\pscirc 42(5)
\pscirc 43(5)
\pscirc 44(5)
\pscirc 45(5)
\pscirc 46(5)
\pscirc 47(5)

\psendfig

%
%
\figvisu{\figBoxA}{Figure 3: The complete exchange model}
{
\figwritew 57:$2b$(4)
\figsetmark{$\figBullet$}
}
\centerline{\box\figBoxA}
\vspace{5mm}

\noindent The motion of a given pair of particles at sites $i, i+1$ is described as the motion of a point particle on a two-dimensional billiard table described by $$\Omega=\left\{(x_1,x_2)\in {\bf R}^2,\, |x_1|\leq b,\, |x_2|\leq b,\,  |x_1-x_2|\leq a\right \}.$$ 
 \figinit{0.3mm}

\figpt 1:(-60,0)
\figpt 2:(40,0)
\figpt 3:(80,40)
\figpt 4:(80,140)
\figpt 5:(-20,140)
\figpt 6:(-60,100)

\figpt 7:(-40,30)
\figpt 8:(-35,55)

\figpt 9:(80,0)
\figpt 10:(-60,140)

\figpt 11:(-10,0)
\figpt 12:(30,140)

\figpt 13:(-60,50)
\figpt 14:(80,90)

\figpt 15:(-80,0)
\figpt 16:(-80,140)
\figpt 17:(-80,70)

\figpt 18:(90,0)
\figpt 19:(90,40)

\figpt 20:(90,20)
\psbeginfig{}

\psline[1,2,3,4,5,6,1]

\pscirc 7(5)
\psset arrowhead(fillmode=yes) \psarrow[7,8]

\psset arrowhead(fillmode=yes) \psarrow[15,16]
\psset arrowhead(fillmode=yes) \psarrow[16,15]

\psset arrowhead(fillmode=yes) \psarrow[18,19]
\psset arrowhead(fillmode=yes) \psarrow[19,18]

\psset (dash=8)

\psline[2,9,3]
\psline[5,10,6]

\psendfig
%
%
\figvisu{\figBoxA}{Figure 4: Billiard giving the motion of two particles in the complete exchange model}
{
\figwritew 17:$2b$(5)
\figwritee 20:$2b-a$(2)
\figsetmark{$\figBullet$}
}
\centerline{\box\figBoxA}
\vspace{5mm}

   It is straightforward to build analogous models on
higher-dimensional lattices.  Those models have been introduced in \cite{Prosen} and called the {\em complete exchange models} \cite{GL}.

\vspace{3mm}

\noindent We denote by $\und \bp=(\bp_1,\ldots,\bp_N)$ and $\und\bq=(\bq_1,\ldots,\bq_N)$ the vectors made of the momenta and positions of each particles.  The evolution of a probability distribution $\bP_t(\und\bp,\und\bq)$ over phase space is given by 
\begin{eqnarray}
\left\{\begin{array}{ll}\partial_t \bP_t=-\sum_{i=1}^N\bp_i\partial_{\bq_i}\bP_t,\;{\rm if }\;\und\bq\in\Omega \\
\bP_t(\bq_1,\ldots,\bq_N,\bp_1,\ldots,\bp_N)=0 \;{\rm if }\;\und\bq\notin\Omega 
\end{array}
\right.
\end{eqnarray}
with specular boundary conditions:
$$
\bP_t(\und\bq,\und \bp)=\bP_t(\und\bq,\und \bp-2(\bn\cdot\und\bp)\bn),\quad {\rm if } \quad \und\bq\in\partial\Omega
$$
and $\bn$ is a normal vector to the boundary at point $\und\bq$.
We denote by  $f_i(\mathbf{q}, \mathbf{p},
t)$, the marginal probability
distribution of the $i$-th particle
\begin{equation}
f_i( \mathbf{q},\mathbf{p},t)=\int\prod_{j\neq i} \ud \bp_j\ud \bq_j\;\bP_t(\bq_1,\ldots,\bq_N,\bp_1,\ldots,\bp_N)
\end{equation}
and similarly $f_{i,j}(\mathbf{q}, \mathbf{p},  \mathbf{q}',\mathbf{p}'
t)$ denotes the probability distribution relative to the momenta and positions of particles $i$ and $j$. The equation for the evolution of the set of probability density of a single particle is given in terms of collision with the boundaries of the walls of its own cell and with its neighbors
\begin{equation}
\frac{d}{dt}f_i( \mathbf{q},\mathbf{p}, t) =
-\mathbf{p}\cdot\nabla_\mathbf{q} f_i + L^\mathbf{w}f_i + L^\mathrm{c}
(f_{i,i+1}) + L^\mathrm{c} (f_{i, i-1}).
\label{Bol}
\end{equation}
Here $L^\mathrm{w}$ accounts for the collisions of the particles with the
walls of  
their respective cells:
\begin{equation}
L^\mathrm{w}(f_i)(\bq,\bp,)=\delta_{\Lambda_V}(\bq)(\bp\cdot\bn)^+[f_i(\bq,\bp-2(\bp\cdot\bn)\bn)-f_i(\bq,\bp)],
\end{equation}
here and below $\bn$ denote a generic unit normal vector to the boundary $\Lambda_U$ or $\Lambda_V$ at point $\bq$.  We use the notation
$$
x^\pm=\left\{\begin{array}{ll} |x|\quad {\rm if}\quad \pm x\geq0 \\
0 \quad {\rm if}\quad \pm x<0.\\
\end{array}
\right.
$$
The collision term $L^\mathrm{c}(f_{i,i\pm 1})$ for the collisions of
the $i$-th particle with the $i\pm1$th is given by,   
\begin{equation}
 L^\mathrm{c}(f_{i,i\pm 1}) 
    = \int \ud \mathbf{p}_a \ud \mathbf{q}'
    \delta_{\Lambda}(\mathbf{q} - \mathbf{q}')
    (p^\bot-p_a^\bot)^+\label{Lc} [
  f_{i,i\pm 1}(\mathbf{q},\mathbf{p}_c,  \mathbf{q}',\mathbf{p}_b) - f_{i,i\pm 1}( \mathbf{q},\mathbf{p},\mathbf{q}',\mathbf{p}_a)]
\end{equation}
with $p_b^\bot=p^\bot$, $p_c^\bot=p_a^\bot$, $\mathbf{p}_c - p_c^\bot 
\mathbf{n} = \mathbf{p} - p^\bot \mathbf{n}$, 
and $\mathbf{p}_b - p_b^\bot \widehat{\mathbf{n}} = \mathbf{p}_a - p_a^\bot
\widehat{\mathbf{n}}$.  
One can check that the distribution 
\begin{equation}
  \bP_\mathrm{eq}(\bq_1,\ldots,\bq_N,\bp_1,\ldots,\bp_N) \equiv Z^{-1}{\bf 1}_{\Omega}(\bq_1,\ldots,\bq_N)
  \prod_{i=1}^Ne^{-\beta \,\bp^2_i/2} 
  \label{mueq}
\end{equation}
is stationary for any inverse temperature $\beta$.  We will explain at the end of the next subsection how to thermalize those systems at different temperatures at their boundaries.

\vspace{3mm}
\noindent{\it Example.} Coming back to the example of the complete exchange models, we see that the  different collision terms take the form
\begin{equation}
L^\mathrm{w}(f_i)(q,p)=\delta_{\pm b}(\bq) p^{\pm}[f_i(q,-p)-f_i(q,p)],
\label{colexchange0}
\end{equation}
\begin{equation}
L^\mathrm{c}(f_{i,i\pm 1}) 
    = \int \ud p' \ud q'
    \delta_{\pm a}(q - q')
    (p-p')^{\pm} [
  f_{i,i\pm 1}( q,p', q',p) - f_{i,i\pm 1}( q,p,q',p')].
\label{colexchange}
\end{equation}

\vspace{3mm}
\noindent As our goal is to see how to approximate those dynamics by stochastic ones. We now describe in detail the case of a free particle confined between two hot walls.  The action of the walls is purely random and express the effect of  the contact of the particle with a very large system whose temperature is fixed and constant.
\subsection{A free particle confined between two thermal walls.}
\label{subsec: confined}
Let us consider  a particle moving in the interval $[0,1]$ and reflected at the left and right boundaries boundaries with a random velocity $p$ whose absolute value is distributed according to distributions respectively $\phi_{\beta_L}$ and $ \phi_{\beta_R}$, with:
\begin{equation}
\phi_{\beta}(dp)=\beta\, p\,e^{-{\beta} \frac{p^2}{2}}\, \bo{(p>0)}\, dp.
\end{equation}
The temperature of the right, respectively left, wall is $T_R:=1/\beta_R$,
resp. $T_L:=1/\beta_L$. As the particle is reflected at the boundary 
of $[0,1]$, we understand that the distribution of velocities at the thermal walls has a sign which is opposite to the sign of the velocity of 
the incoming particle. 

In the following, we recall some notation and results following \cite{LefevereZambotti1}. Let $E=\{-1,+1\}$, $(q_0,p_0)$ the initial data of the particles, and define $\sigma_0= \mathrm{sign}(p_0)$, $\sigma_k=(-1)^k \sigma_0$ for $k \ge 0$. For $\sigma \in E$. In  \cite{LefevereZambotti1} a more general case has been addresed, where $(\sigma_k)$ is Markov chain on the state space $E$.

As explained in \cite{LefevereZambotti1}, this is a Markov process, the notations are similar to \cite{LefevereZambotti1}. We explain here the special case corresponding to the case we are interested in.
The state space of the associated Markov chain is $E=\{-1,+1\}$.
Let $(q_{0},p_{0})$ the initial data and velocity of the particle.  We define $\sigma_{0}={\rm sign}(p_{0})$.
We consider now the Markov chain $(\sigma_{k})_{k\geq 0}$ in $E$ with initial
state $X_0=\sigma_{0}$. In fact, the Markov chain has a deterministic evolution $\sigma_{k}=(-1)^k\sigma_{0}$, $k\geq 0$.

For each $\sigma\in E$, we write $\hat\sigma=\hf(\sigma+1)$.
Then the time of the first collision with a wall is
\[
S_{0}=S_{0}(q_{0},p_{0}):=\frac{\hat\sigma_{0}-q_{0}}{p_{0}}>0,
\]
We now define the sequence of times the particle takes between two
subsequent visits to the scatterers.
Conditionally on $(\sigma_{k})_{k\geq 0}$,
the sequence $(\tau_{k})_{k\geq 1}$
is independent with distribution defined by
\begin{equation}\label{taucon}
 \bbP(\tau_{k}\in d\tau \, | \, \sigma_{k-1})
= \frac{\beta_{\sigma_{k-1}}}{\tau^3} \,
\exp\left(-\frac{\beta_{\sigma_{k-1}}}{2\tau^2}\right)
\, \bo{(\tau>0)} \, d\tau,
\end{equation}
where $\beta_{-1}=\beta_L$ and $\beta_1=\beta_R$. In other words,
the conditional law of $1/\tau_k$ is $\phi_{\beta_{\sigma_{k-1}}}$:
\begin{equation}\label{1/taucon}
 \bbP(1/\tau_{k}\in dp \, | \, \sigma_{k-1})
= \beta_{\sigma_{k-1}}\, p\,e^{-{\beta_{\sigma_{k-1}}} \frac{p^2}{2}}\, \bo{(p>0)}\, dp.
\end{equation}
The time of the $(k+1)$-st collision with one of the two walls is
\begin{equation}\label{rene}
S_{k}:=S_{0}+\tau_{1}+\cdots+\tau_{k}, \qquad k\geq 1.
\end{equation}
Notice that $(S_k)_{k\geq 0}$ is a standard renewal process,
see \cite{Asmussen}.
Before time $S_{0}$, the particle moves with uniform velocity
$p_{0}$. Between time $S_{k-1}$ and time $S_{k}$, the particle
moves with uniform velocity $\frac{\sigma_{k}}{\tau_{k}}$ and
$(S_{k})_{k\geq 0}$ is the sequence of times when $q_{t}\in
\{0,1\}$. In particular we define the sequence of incoming velocities
$v_{k}$ at time $S_{k}$
\begin{equation}\label{vcon}
v_{0}:=p_{0}, \qquad v_{k}:=\frac{\sigma_{k}}{\tau_{k}}, \quad k\geq 1.
\end{equation}
We define the stochastic process $(q_{t},p_{t})_{t\geq 0}$
with values in $[0,1]\times \R^*$
\begin{equation}\label{qpcon}
(q_t,p_t):=
\left\{
\begin{array}{ll}
(q_{0}+p_{0}t,p_{0}) \qquad {\rm if} \quad  t<S_{0},
\\ \\
\left( \hat\sigma_{k-1}
+\frac{\sigma_{k}}{\tau_{k}}(t-S_{k-1}),
\frac {\sigma_{k}}{\tau_{k}}
\right)  \ \ {\rm if} \ \
S_{k-1}\leq t< S_{k}, \ k\geq 1,
\end{array}
\right.
\end{equation}
\noindent 
The energy exchanged between the two walls
during a time interval $[0,t]$ is given by
\[
J[0,t]:=\hf\sum_{k\geq 1:\, S_{k}\leq t}
v^2_{k} \, \sigma_{k}.
\]
and we have shown \cite{LefevereZambotti1} that
\begin{equation}
\lim_{t\rightarrow\infty}\frac 1 t J[0,t]=\kappa (T_L-T_R)
\label{defcond}
\end{equation}
where
$T_R=1/\beta_R$, $T_L=1/\beta_L$ and $\kappa$ is by definition the {\it conductivity} of the model, given by:
\begin{equation}
\kappa^{-1}=\left(\frac{\pi\beta_L}{2}\right)^\hf+\left(\frac{\pi\beta_R}{2}\right)^\hf.
\label{conductivity}
\end{equation}

\noindent We will establish a correspondence between the deterministic dynamics and stochastic dynamics by using simplification of the evolution equation of the probability distributions.  Thus, we need to write the infinitesimal evolution of probability distributions under the stochastic dynamics described here.  We set for bounded Borel $f:[0,1[\times \R_+\mapsto\R$
\[
P_tf(q_0,p_0):= \E(f(q_t,p_t)) = \E(f(F(q,p,t,(\tau_n)_{n\geq 1})), \qquad (q_0,p_0)\in[0,1]\times \R^*.
\]
In the appendix, we prove the following:
\begin{prop}\label{infinitesimal2}
For all $f,g:[0,1]\times\R_+\mapsto\R$ bounded with bounded
continuous first derivatives:
\[
\begin{split}
& \left. \frac d{dt} \, \int_0^1 dq\, \int_{\R_+}dp \,
g(q,p)\, P_tf(q,p) \right|_{t=0} =
\\ & = \int_{\R_+} dp \, \int_0^1 dq \, g(q,p)\, p \,
f_q(q,p) + \int_{\R_+} dp \, p\, g(1,p) \int_{\R_+}
\phi_{\beta_R}(du)\, (f(1,-u)-f(1,p))\\
& + \int_{\R_-} dp \, p\, g(0,p) \int_{\R_+}
\phi_{\beta_L}(du)\, (f(0,u)-f(0,p)).
\end{split}
\]
\end{prop}
\noindent
From this computation we obtain a formal expression for the generator:
\begin{equation}\label{generator}
\begin{split}
Lf(q,p)= p\, \frac{\partial f}{\partial q}& + p^-\delta_0(dq) \left[ \int_{\R_+}
\phi_{\beta_L}(du)\, (f(0,u)-f(0,p))\right]\nonumber \\
&+  p^+\delta_1(dq) \left[ \int_{\R_+}
\phi_{\beta_R}(du)\, (f(1,-u)-f(1,p))\right].
\end{split}
\end{equation}
And similarly, one may obtain an expression for the formal adjoint:
\begin{eqnarray}
L^*g(q,p)= -p\, \frac{\partial g}{\partial q}& -& \delta_0(dq) \left[
p^+\, g(0,p) - \phi^+_{\beta_L}(p) \int_{\R_-} du\, u\, g(0,u) \right]\nonumber \\
&-& \delta_1(dq) \left[
p^-\, g(1,p) - \phi^-_{\beta_R}(p) \int_{\R_+} du\, u\, g(1,u) \right].
\label{adjoint2}
\end{eqnarray}
This tells us also how to describe the thermalization on the boundaries of the deterministic dynamics described in the previous subsection.  For instance, in the case of a system of $N$ particles described by the complete exchange dynamics described in (\ref{colexchange0}) and (\ref{colexchange}), one simply replaces 
$L^\mathrm{w}(f_1)$ and $L^\mathrm{w}(f_N)$ by:
\begin{equation}
L^{\beta_L}(f_1)(q,p)=-\delta_{\pm b}(dq)\left[p^{\pm}f_1(q,p)-\phi^{\pm}_{\beta_L}(p)\int_{\R_{\mp}} du\, u\,f_1(q,u)\right ]
\label{therwall1}
\end{equation}
and
\begin{equation}
L^{\beta_R}(f_N)(q,p)=-\delta_{\pm b}(dq)\left[p^{\pm}f_N(q,p)-\phi^{\pm}_{\beta_R}(p)\int_{\R_{\mp}} du\, u\,f_N(q,u)\right ]
\label{therwall2}
\end{equation}
where $\phi^{\pm}_\beta(p)=\beta(p)^\pm e^{-\beta\frac{p^2}{2}}$.
In more general models, the thermalization is similar: the specular reflections of the particle on the walls of its own cell is replaced by the action of a thermal wall.

We notice here that our process is {\it piecewise deterministic}, at least in the sense that randomness acts in a discrete (random) set of times and the motion is deterministic in between. Piecewise deterministic processes have been extensively studied, see for instance \cite{davis, jacobsen}, however our case does not fit in the standard framework. Indeed, in the literature one finds piecewise deterministic processes with generators like \eqref{therwall1} and
\eqref{therwall2} with the measures $\delta_{\pm b}(dq)$ replaced by some function on the state space; this corresponds to a noise which can act at a positive and finite rate all over the state space, while in our situation the noise
act every time that, and only when, the process hits the two lines $\{(q,p): q=\pm 1\}$. Therefore the standard theory can not be applied to our processes
\eqref{qpcon}.

\subsection{Stochastic approximation of the deterministic dynamics.} We want to describe the dynamics of a given particle when its neighbors have positions and velocities distributed according to a local equilibrium distribution. In order to do so, in the expression (\ref{colexchange}), we set:
\begin{equation}
f_{i,i\pm 1}(q,p,q',p')=f_i(q,p)\sqrt{\frac{\beta_{i\pm 1}}{8b^2\pi}}e^{-\beta_{i\pm 1}\frac{p'^2}{2}}\bo{[-b,b]}(q').
\end{equation}
Then (\ref{colexchange}) becomes
\begin{equation}
\begin{split}
L^\mathrm{c}(f_{i,i\pm 1})(q,p) =&\sqrt{\frac{\beta_{i\pm 1}}{8b^2\pi}}\left [\bo{[-b,b-a]}(q)\int \ud p'(p-p')^{+}[f_i(p',q)e^{-\beta_{i\pm 1}\frac{p^2}{2}}-f_i(p,q)e^{-\beta_{i\pm 1}\frac{p'^2}{2}}]\right. \\
&+\left.\bo{[a-b,b]}(q)\int \ud p'(p-p')^{-}[f_i(p',q)e^{-\beta_{i\pm 1}\frac{p^2}{2}}-f_i(p,q)e^{-\beta_{i\pm 1}\frac{p'^2}{2}}]\right].
\end{split}
\end{equation}
We recall that $b<a<2b$.  Heuristically, the interpretation of the dynamics  with the collision term modified as above is easy to describe: each particle moves freely and bounces back and forth in its cell of size $2b$ except when it enters two ``critical" regions of size $2b-a$ located near the boundaries of the interval $[-b,b]$.  There, at a random point it may collide elastically with a particle whose velocity is Maxwellian.  Simplifying further, one may contract the critical regions to the points at the boundaries of the interval and replace the collision with a Maxwellian particle by the collision with a thermal wall.  One is then lead to equations of the form (\ref{therwall1}), (\ref{therwall2}).  One may of course use the same strategy for any of the deterministic collisional dynamics we described above. The geometry may be completely different but the idea is always similar: each particle moves freely in its cell except in a ``critical region" where collisions with neighbors may occur.  If one is interested in the  local equilibrium dynamics, one is led to a dynamics described by free motion and interactions with  thermal walls.  Pictorially, the model of figure 2 is transformed into the model of figure 5 below.
\vspace{3mm}
\figinit{pt}
\figpt 1:(-210, 60)
\figpt 2:(-210, 0)
\figpt 3:(210, 0)
\figpt 4:(210,60)

\figpt 50:(-160,60)
\figpt 51:(-140,60)
\figpt 60:(-100,60)
\figpt 61:(-80,60)
\figpt 70:(-40,60)
\figpt 71:(-20,60)
\figpt 80:(20,60)
\figpt 81:(40,60)
\figpt 90:(80,60)
\figpt 91:(100,60)
\figpt 10:(140,60)
\figpt 11:(160,60)

\figpt 150:(-160,0)
\figpt 151:(-140,0)
\figpt 160:(-100,0)
\figpt 161:(-80,0)
\figpt 170:(-40,0)
\figpt 171:(-20,0)
\figpt 180:(20,0)
\figpt 181:(40,0)
\figpt 190:(80,0)
\figpt 191:(100,0)
\figpt 110:(140,0)
\figpt 111:(160,0)

\figpt 211:(-150,30)
\figpt 212:(-90,30)
\figpt 213:(-30,30)
\figpt 214:(30,30)
\figpt 215:(90,30)
\figpt 216:(150,30)


\figpt 29:(-180,35)
\figpt 291:(-165,20)

\figpt 30:(-130,10)
\figpt 301:(-115,35)

\figpt 31:(-50,45)
\figpt 311:(-65,30)

\figpt 32:(-10,20)
\figpt 321:(10,35)

\figpt 33:(57,15)
\figpt 331:(70,35)

\figpt 34:(135,31)
\figpt 341:(120,10)

\figpt 35:(190,41)
\figpt 351:(170,20)

\psbeginfig{}
\psline[1,2,3,4,1]
\pssetfillmode{yes}\pssetgray{0.7}
\psline[150,151,51,50,150]
\psline[160,161,61,60,160]
\psline[170,171,71,70,170]
\psline[180,181,81,80,180]
\psline[190,191,91,90,190]
\psline[110,111,11,10,110]

\pssetfillmode{no}\pssetgray{0}
\pscirc 29(6)
\pscirc 30(6)
\pscirc 31(6)
\pscirc 32(6)
\pscirc 33(6)
\pscirc 34(6)
\pscirc 35(6)

\psset arrowhead(fillmode=yes) \psarrow[29,291]
\psset arrowhead(fillmode=yes) \psarrow[30,301]
\psset arrowhead(fillmode=yes) \psarrow[31,311]
\psset arrowhead(fillmode=yes) \psarrow[32,321]
\psset arrowhead(fillmode=yes) \psarrow[33,331]
\psset arrowhead(fillmode=yes) \psarrow[34,341]
\psset arrowhead(fillmode=yes) \psarrow[35,351]

\psendfig
%
\figvisu{\figBoxA}{Figure 5: Collisional regions have been replaced by interactions with thermal walls}
{\figwritec [211]{$T_1$}
\figwritec [212]{$T_2$}
\figwritec [213]{$T_3$}
\figwritec [214]{$T_4$}
\figwritec [215]{$T_5$}
\figwritec [216]{$T_6$}
}
\centerline{\box\figBoxA}
\vspace{3mm}
Clearly, the direction orthogonal to the array of heat baths is irrelevant to the transport of energy and one is thus lead to the model of {\it confined tracers} introduced in \cite{LefevereZambotti1}. In other words,  one replaces the effect of the neighbors of a given particle by the effect of heat baths.  The dynamics of each particle is given by the motion of a particle confined between two thermal walls as in subsection \ref{subsec: confined}.  We discuss now the large deviations properties of the current carried by a single particle confined between two hot walls.

\section{Fluctuations and large deviations of the current.}
Our final goal is to study a suitable scaling limit of the large deviations functional of the energy current in the case of a confined tracer.  In order to identify the different physical quantities appearing in the scaling limit, we start off by computing the variance of the current in equilibrium
$$
\lim_{t\to\infty}\frac{1}{t}\E_{{\rm eq}}\left((J[0,t])^2\right)
$$
where the expectation is taken with respect to the process defined in subsection \ref{subsec: confined} with $\beta_L=\beta_R=\beta=T^{-1}$ for some $T>0$.
In principle this quantity should be computable as the second derivative with respect to $\lambda$ of the cumulant generating function
\begin{equation}
f(\lambda):=
\lim_{t\to+\infty} \frac1t \, \log\E_{{\rm eq}}\left(\exp\left(\lambda J[0,t]\right)\right),
\end{equation}
Unfortunately, this amounts to  assuming that one can exchange the limit $t\to\infty$ and the derivatives with respect to $\lambda$.  As we are not aware of any general argument justifying this exchange we compute explicitly the variance of the current below.  We find that it coincides indeed with the second derivative of the cumulant generating function obtained in \cite{LefevereZambotti1}. 
We will see that the variance (\ref{variance}) appears in (\ref{iepsilon}).

\subsection{Variance of the current in equilibrium.}  
We recall that we can associate to the renewal process $(S_k)_{k\geq 0}$ the  counting process
\begin{equation}\label{reneco}
N_t:=\sum_{k=1}^{\infty} \bo{(S_{k}\leq t)}.
\end{equation}
\begin{prop} The Green-Kubo relation 
\begin{equation}
\sigma(T):=\lim_{t\to\infty}\frac{1}{t}\E_{{\rm eq}}\left((J[0,t])^2\right)=4\kappa(T) T^2
\label{variance}
\end{equation}
holds true, where $\kappa(T)$ is given by (\ref{conductivity}) when $\beta_L^{-1}=\beta_R^{-1}=T$
\end{prop}

\begin{proof}
 We recall that the energy
exchanged between the two walls
during a time interval $[0,t]$ is given by

\[
J[0,t]=\hf\sum_{k\geq 1:\, S_k\leq t}\sigma_0(-1)^k v^2_{k}
=\hf\sum_{k=1}^{N_t} \sigma_0(-1)^{k}v^2_{k}\, .
\]
Write
\[
4\left( J[0,t] \right)^2 = \left(
\sum_{k=1}^{N_t} (-1)^{k}v^2_{k} \, \right)^2.
\]
Then
\[
\begin{split}
& 4\, \E_{\rm eq}\left( \left(J[0,t]\right)^2\right)
= \E_{\rm eq}\left(\sum_{k=1}^{N_t} (-1)^{k} v^2_k
\sum_{k'=1}^{N_t} (-1)^{k'} v^2_{k'}\right)
\\ & =
\sum_{\ell=1}^\infty \sum_{1\leq k,k'\leq \ell}
\E_{\rm eq}\left( \bo{(N_t=\ell)} \, (-1)^{k+k'} v^2_{k} \, v^2_{k'} \right)
\\ & =
\sum_{\ell=1}^\infty \sum_{k=1}^\ell
\E_{\rm eq}\left( \bo{(N_t=\ell)}  \, v_{k}^4  \right) +
2\sum_{\ell=2}^\infty \sum_{1\leq k<k'\leq \ell}
\E_{\rm eq}\left( \bo{(N_t=\ell)}  \, (-1)^{k+k'} v^2_{k} \, v^2_{k'} \right).
\end{split}
\]
Observe now that conditionally on $(N_t=\ell)$, the variables
$\{v^2_k, k\leq \ell\}$ are exchangeable. Therefore
\[
\sum_{\ell=1}^\infty \sum_{k=1}^\ell \E_{\rm eq}\left( \bo{(N_t=\ell)}  \, v_{k}^4 \right) =
\sum_{\ell=1}^\infty \ell\ \E_{\rm eq}\left( \bo{(N_t=\ell)}  \, v_{1}^4 \right)
= \E_{\rm eq}\left( N_t  \, v_{1}^4 \right).
\]
Moreover, by exchangeability
\[
\begin{split}
& \sum_{\ell=2}^\infty \sum_{1\leq k<k'\leq \ell}
\E_{\rm eq}\left( \bo{(N_t=\ell)}  \, (-1)^{k+k'} v^2_{k} \, v^2_{k'} \right)
\\ & = \sum_{\ell=2}^\infty \sum_{1\leq k<k'\leq \ell} (-1)^{k+k'} \,
\E_{\rm eq}\left( \bo{(N_t=\ell)}  \, v^2_{1} \, v2_{2} \right)
\\ & = \sum_{\ell=2}^\infty (-1)^{\lfloor \ell/2\rfloor} \,
\E_{\rm eq}\left( \bo{(N_t=\ell)}  \, v^2_{1} \, v^2_{2} \right)
= \E_{\rm eq}\left( (-1)^{\lfloor N_t/2\rfloor}  \, v^2_{1} \, v^2_{2} \right)
\end{split}
\]
and therefore
\[
4\E_{\rm eq}\left( \left(J[0,t]\right)^2\right) =
\E_{\rm eq}\left( N_t  \, v_{1}^4 + 2\, (-1)^{\lfloor N_t/2\rfloor}  \, v^2_{1} \, v^2_{2} \right).
\]

We now compute
\begin{equation}
\begin{split}
& \lim_{t\to+\infty} \frac1t \, \E_{\rm eq}\left( \left(J[0,t]\right)^2\right) =
\lim_{t\to+\infty} \frac1{4t}\, \E_{\rm eq}\left( N_t  \, u_{1}^4 + 2\, (-1)^{\lfloor N_t/2\rfloor}  \, v^2_{1} \, v^2_{2} \right)\label{JNt}
\\ & =
\frac1{4\mu}  \, \E_{\rm eq}\left( u_{1}^4\right)
\end{split}
\end{equation}
where
\[
\mu =  \beta\int_0^\infty e^{-\beta\frac{v^2}2} \, dv = \sqrt{\frac{\beta\pi}2}
\]
and the last equality follows because the second term is uniformly bounded in $t$ and because
\begin{equation}
\lim_{t\to+\infty}\E_{\rm eq}\left(\frac{N_t}{t}-\frac1\mu\right)^2=0.
\label{vari0}
\end{equation}
In order to see this last point, it is enough to prove that ${\rm Var}(\frac{N_t}{t})\to 0$, since we already now that $N_t/t\to 1/\mu$ in $L^1$ by the Renewal theorem. We approximate the variables $(\tau_k)_{k\in\bbN}$ by a sequence of truncated variables $(\tau^n_k)_{k\in\bbN}$, $\tau^n_k=\tau_k\wedge n$. Each $\tau^n_k$ has a finite variance, say $\sigma_n^2$  and average $\mu_n$ and theorem V.6.3 in \cite{Asmussen} implies for the corresponding renewal process $N^n_t$ that
$$
\lim_{t\to+\infty}\frac 1 {t} {\rm Var}\left(N_t^n\right)=
\frac{\sigma^2_n}{\mu_n^3}
$$
and thus for each $n$,
$$
\lim_{t\to+\infty} \E_{\rm eq}\left(\frac{N^n_t}{t}-\frac1{\mu_n} \right)^2=0.
$$
Notice now that $\tau^n_k\leq \tau_k$ a.s. for all $n,k$ implies $N_t\leq N_t^n$ a.s. and therefore
\[
{\rm Var}\left(N_t\right) = \E\left((N_t)^2\right) - 
(\E\left(N_t\right))^2 \leq {\rm Var}\left(N_t^n\right) + (\E\left(N_t^n\right))^2
- (\E\left(N_t\right))^2.
\]
Then for all $n\geq 1$
\[
\varlimsup_{t\to+\infty}\frac 1 {t^2} {\rm Var}\left(N_t\right)\leq
\frac1{\mu_n^2}-\frac1{\mu^2}.
\]
Since $\mu_n=\E(\tau^n_1)\uparrow\E(\tau_1)=\mu$ by monotone convergence, then
we obtain that indeed ${\rm Var}(\frac{N_t}{t})\to 0$
and we have proven \eqref{vari0}.
Now
\[
\E(v_1^4) = \beta\int_0^\infty v^5 \, e^{-\beta\frac{v^2}2} \, dv
= \frac{8}{\beta^2}.
\]
Therefore
\begin{equation}\label{n=m}
\lim_{t\to+\infty} \frac1t \, \E_{\rm eq}\left( \left(J[0,t]\right)^2\right) =
\frac{2}{\beta^2}\sqrt{\frac2{\beta\pi}}=4T^2\kappa(T).
\end{equation}
\end{proof}
\subsection{Cumulant generating function}
\label{large}
We state first the results of \cite{LefevereZambotti1} which apply to the case of a single tracer confined between to thermal walls as described in subsection \ref{subsec: confined}.
We define $\cB:=(\beta_L,\beta_R)$ and

\begin{equation}\label{freeener}
f(\lambda,\beta_L,\beta_R)=f(\lambda,\cB):=
\lim_{t\to+\infty} \frac1t \, \log\E\left(\exp\left(\lambda J[0,t]\right)\right),
\qquad  \lambda\in\, ]-\beta_R,\beta_L[,
\end{equation}
\begin{equation}\label{calpha}
C(x,\eta):=
\int_0^{+\infty} v\, e^{- \frac\eta v -x\, \frac{v^2}2 } \, dv,
\qquad \eta, x\geq 0,
\end{equation}
and 
\begin{equation}\label{F}
F(\lambda,\eta,\cB):=\beta_L\beta_R\, C(\beta_R+\lambda,\eta)
\, C(\beta_L-\lambda,\eta).
\end{equation}
From \cite{LefevereZambotti1}, we have the
\begin{prop}\label{mgf} Suppose that $\beta_{L}\leq \beta_{R}$. The function $f(\cdot,\cB)$ is convex and continuous
over $]-\beta_R,\beta_{L}[$ and satisfies the Gallavotti-Cohen symmetry relation
\begin{equation}\label{GC}
f(\lambda,\cB)=f(\beta_{L}-\beta_R-\lambda,\cB).
\end{equation}
$f(\cdot,\cB)$  is analytic on $
]-\beta_{R},\beta_{L}-\beta_R[ \, \cup \, ]0,\beta_L[$.  Moreover
\begin{enumerate}
\item $\forall \,\lambda\in\,]-\beta_{R},\beta_{L}-\beta_R[ \, \cup \, ]0,\beta_L[$, $f(\lambda,\cB)$ is given by the unique solution
$\eta_0>0$ to the equation
\begin{equation}
F(\lambda,\eta_0,\cB)=1.
\label{F1}
\end{equation}
\item If $\lambda\in \, [\beta_{L}-\beta_R,0]$, then $f(\lambda,\cB)=0$.
\item If $\lambda\notin \, ]-\beta_R,\beta_{L}[$ then $f(\lambda,\cB)=+\infty$.
\item $\frac{\partial f}{\partial\lambda}(0^+,\cB)=\kappa(T_L-T_R)$
\item $\frac{\partial f}{\partial\lambda}:\, ]0,\beta_L[\,\mapsto\,]\kappa(T_L-T_R),+\infty[$ and $\frac{\partial f}{\partial\lambda}:\, ]-\beta_{R},\beta_{L}-\beta_R[\,\mapsto\,]-\infty,-\kappa(T_L-T_R)[$ are bijections.
\end{enumerate}
\end{prop}
\begin{proof}
The result is proved in \cite[Propositions 4.5, 4.7]{LefevereZambotti1} (notice however that
there is a difference of sign, i.e. in \cite{LefevereZambotti1} one finds statements and proofs are given for $f(-\lambda,\cB)$). It only remains to prove point (5). We consider $\lambda\in\, ]0,\beta_L[$. Then by \eqref{GC},
$f_n(\lambda,\cB)>0$ is defined by the relation $F(\lambda,f_n(\lambda,\cB),\cB)=1$.
Hence by the implicit function Theorem
\[
\frac{\partial f}{\partial \lambda}(\lambda,\cB)
= - \frac{\partial F}{\partial \lambda}(\lambda,f(\lambda,\cB),\cB)
\left( \frac{\partial F}{\partial\epsilon}(\lambda,f(\lambda,\cB),\cB)\right)^{-1} .
\]
A computation yields
\[
\begin{split}
 & \frac{\partial F}{\partial \lambda}(\lambda,\eta,\cB) = 
 \\ & = \beta_R\beta_L\int_{\R_+^2} v_1\, v_2\,
 \frac12\left(v_2^2-v_1^2\right)
\exp\left(-\frac{\eta}{v_1} -\frac{\eta}{v_2}-\left(\beta_R+\lambda\right)\, \frac{v_1^2}2
-\left(\beta_L-\lambda\right)\, \frac{v_2^2}2\right) dv_1 \,dv_2,
\end{split}
\]
and
\[
\begin{split}
 & 
 \frac{\partial F}{\partial \eta}(\lambda,\eta,\cB)
 =  \\ & = 
 \beta_R\beta_L\int_{\R_+^2} 
 \left({v_1}+{v_2}\right)
\exp\left(-\frac{\eta}{v_1} -\frac{\eta}{v_2}-\left(\beta_R+\lambda\right)\, \frac{v_1^2}2
-\left(\beta_L-\lambda\right)\, \frac{v_2^2}2\right) dv_1 \,dv_2.
\end{split}
\]
Since $f(0,\cB)=0$, by letting $\lambda\downarrow 0$ we find
\[
\begin{split}
\frac{\partial f}{\partial \lambda}(0^+,\cB) & =
\frac{\beta_R\beta_L\int_{\R_+^2} v_1 \, v_2 \,
 \frac12\left(v_2^2-v_1^2\right)
e^{-\beta_R\frac{v_1^2}2
-\beta_L\frac{v_2^2}2} \, dv_1 \, dv_2}
{\sum_{i\in\{L,R\}} \beta_i\int_{\R_+}
\exp\left(-\beta_i\frac{v^2}2\right) dv}
\\ & =\frac{\frac1{\beta_L}-\frac1{\beta_R}}{\sqrt{\frac{\pi\beta_L}2}+
\sqrt{\frac{\pi\beta_R}2}} = \kappa(T_L-T_R).
\end{split}
\]
On the other hand, as $\lambda\uparrow \beta_L$ we find with a change of variable that
\[
\frac{\partial F}{\partial \lambda}(\lambda,f(\lambda,\cB),\cB)
\sim C_1 (\beta_L-\lambda)^{-2}, \qquad
\frac{\partial F}{\partial \eta}(\lambda,f(\lambda,\cB),\cB)
\sim C_2 (\beta_L-\lambda)^{-1},
\]
for some positive constants $C_1,C_2$, and therefore 
\[
\frac{\partial f}{\partial \lambda}(\lambda,\cB)\sim \frac{C_1}{C_2} \, (\beta_L-\lambda)^{-1}, \qquad \lambda\uparrow \beta_L.
\]
In order to prove that $\frac{\partial f}{\partial\lambda}:\, ]-\beta_{R},\beta_{L}-\beta_R[\,\mapsto\,]-\infty,-\kappa(T_L-T_R)[$ is a bijection, it is enough to apply the Gallavotti-Cohen symmetry relation \eqref{GC}.
\end{proof}

\subsubsection{The Legendre transform of the cumulant generating function}
We  define the Legendre transform of $f(\cdot,\cB)$
\begin{equation}
I(j,\cB)=\sup_{\lambda}\{\lambda j-f(\lambda,\cB)\},
\qquad j\in\R.
\label{Ij}
\end{equation}
We have the following simple
\begin{lemma} The function $I(\cdot,\cB)$ is positive, convex and finite on $\R$.
Setting $j^*:=\kappa(T_L-T_R)$,
\begin{enumerate}
\item $I(\cdot,\cB)$ is smooth and strictly convex over $\R\setminus[-j^*,j^*]$
\item $I(j,\cB)=0$ for all $j\in[0,j^*]$
\item $I(j,\cB)=(\beta_R-\beta_L)|j|$ for all $j\in[-j^*,0]$.
\end{enumerate}
Finally, $I(\cdot,\cB)$ satisfies the Gallavotti-Cohen symmetry relation
\begin{equation}
I(j,\cB)=(\beta_L-\beta_R)j+I(-j,\cB), \qquad j\in\R.
\label{GCI}
\end{equation}
\label{GCI0}
\end{lemma}
Qualitatively, $I(\cdot,\cB)$ has the profile pictured in Figure 1.
\begin{proof} 
By Proposition \ref{mgf}(5), for all $j\in\R\setminus[-j^*,j^*]$
there exists $\lambda\in \,
]-\beta_R,\beta_{L}[$ with $\frac{\partial f}{\partial\lambda}(\lambda,\cB)=j$. 
By \cite[Lemma 2.3.9(b)]{demzei}, for such $(j,\lambda)$ we have $I(j,\cB)=j\lambda-f(\lambda,\cB)$ and moreover $\lambda$ is an exposing hyperplane for $j$, yielding
smoothness and strict convexity.

Let us now fix $j\in[-j^*,j^*]$ and consider the function $g(\lambda)=j\lambda-f(\lambda,\cB)$, $\lambda\in\R$. By Proposition \ref{mgf}(5), $g'(\lambda)>0$ for $\lambda\in\, ]-\beta_{R},\beta_{L}-\beta_R[$ and $g'(\lambda)<0$ for $\lambda\in\, ]0,\beta_{L}[$, so that the supremum of $g$ over $\R$ is the same as the supremum of $g$ over $[\beta_{L}-\beta_R,0]$. But on $[\beta_{L}-\beta_R,0]$ we have $f(\cdot,\cB)\equiv 0$, so that $g(\lambda)=j\lambda$. Therefore $I(j,\cB)=g(0)=0$ if $j>0$ and $I(j,\cB)=g(\beta_{L}-\beta_R)=(\beta_R-\beta_L)|j|$ if $j<0$.

The Gallavotti-Cohen symmetry relation \eqref{GC} for $f$ becomes now \eqref{GCI},
since
\[
\begin{split}
I(j,\cB) & = \sup_{\lambda}\{\lambda j-f(\lambda,\cB)\} =
\sup_{\lambda}\{\lambda j-f(\beta_{L}-\beta_R-\lambda,\cB)\}
\\ & = \sup_{\lambda'}\{(\beta_{L}-\beta_R-\lambda') j-f(\lambda',\cB)\}
=(\beta_L-\beta_R)j+I(-j,\cB).
\end{split}
\]
\end{proof}

\subsection{Large Deviations of the current}
The main result of this section is the following
\begin{theorem} Let $T>0$ and $\tau\geq 0$.
The law $\mu_t$ of $J[0,t]/t$ satisfies a large deviations principle as $t\to+\infty$ with speed $t$ and rate $I(j):=I(j,\cB)$, $j\in\R$, i.e. for any Borel set $A\subseteq\R$ we have the \emph{upper bound}
\begin{equation}\label{upper}
\varlimsup_{t\to+\infty}
\frac1t \, \log \mu_t(A) \leq -\inf_{\overline{A}} I,
\end{equation}
and the \emph{lower bound}
\begin{equation}\label{lower}
\varliminf_{t\to+\infty}\frac1t \, \log \mu_t(A)
\geq -\inf_{\mathring{A}} I.
\end{equation}
\end{theorem}
\begin{proof}
With no loss of generality let $\beta_{L}\leq \beta_{R}$. Note that $I(\cdot,\cB)$ is a continuous, convex, coercive functions. Therefore 
the upper bound \eqref{upper} is a direct consequence of \eqref{freeener}, \eqref{GCI} and the G\"artner-Ellis theorem, see \cite[Theorem 
2.3.6(a)]{demzei}. 

We first prove that the lower bound \eqref{lower} holds for all $A\subset \R\setminus [-j^*,j^*]$, where, as in \eqref{conductivity},
\[
j^*=\kappa(T_L-T_R)=\frac{\beta_L^{-1}-\beta_R^{-1}}
{\left(\frac{\pi\beta_L}2\right)^{\frac12}+\left(\frac{\pi\beta_R}2\right)^{\frac12}}.
\]
Indeed, still by the G\"artner-Ellis theorem, see \cite[Theorem 2.3.6(b)]{demzei}, the lower bound holds 
for open subsets $O \subset \R$ such that, for all $j\in O$ there exists $\lambda_j \in \,
]-\beta_R,\beta_{L}[$ with $\frac{\partial f}{\partial\lambda}(\lambda_j,\cB)=j$. 
By Proposition \ref{mgf}-(5), this holds true for all $j \in \R\setminus [-j^*,j^*]$ and thus the lower bound holds for open (or, equivalently, Borel) sets $O\subset \R 
\setminus [-j^*,j^*]$.

In order to prove the full lower bound for all open subsets of $\R$, we will show the following. For each $j\in[-j^*,j^*]$, there exists a 
sequence $(\nu_t)_t$ of probability measures on $\R$ such that $\nu_t \rightharpoonup \delta_j$ and 
\begin{equation}\label{ent}
\limsup_{t\to+\infty}\frac1t \, \H(\nu_t \, | \, \mu_t)\leq I(j).
\end{equation}
Indeed, if \eqref{ent} is proved, then we argue as follows.
Let $j\in[-j^*,j^*]$ and let $O$ be an open neighborhood of $j$. Then, using $\log  \nu_t(O) \le 0$ and Jensen inequality
\[
\begin{split}
\log \mu_t(O) & =\log \int_O \frac{d\mu_t}
{d\nu_t} \, {d\nu_t} = \log \left( \frac1{\nu_t(O)}\int_O
\frac{d\mu_t} {d\nu_t} \, {d\nu_t}\right) + \log \nu_t(O)
\\ & \geq \log \left( \frac1{\nu_t(O)}\int_O
\frac{d\mu_t} {d\nu_t} \, {d\nu_t}\right)
\geq \frac1{\nu_t(O)}\int_O
\log \left( \frac{d\mu_t} {d\nu_t}\right)  \, {d\nu_t}.
\end{split}
\]
Since $x\log x\geq -e^{-1}$ for $x\geq 0$
\[
\begin{split}
\log \mu_t(O) & \geq \frac1{\nu_t(O)}\left( -\H(\nu_t\,|\,\mu_t)+\int_{O^c}
\log \left( \frac{d\nu_t} {d\mu_t}\right)  \, \frac{d\nu_t} {d\mu_t}\, {d\mu_t}\right)
\\ & \geq \frac1{\nu_t(O)}\left( -\H(\nu_t\,|\,\mu_t)-e^{-1}
\right).
\end{split}
\]
Since $j\in O$, $\nu_t\rightharpoonup\delta_j$ and $O$ is open, then 
$\nu_t(O)\to 1$ as $t\to+\infty$. We obtain by \eqref{ent}
\[
\varliminf_{t\to+\infty} \frac1t \, \log \mu_t(O)\geq 
-\varlimsup_{t\to+\infty} \frac1t  \H(\nu_t \,| \, \mu_t)\geq -I(j).
\]
Therefore, optimizing over $j\in [-j^*,j^*]$
\[
\varliminf_{t\to+\infty} \frac1t \, \log \mu_t(O)
\geq - \inf_{j\in O \cap [-j^*,j^*]} I(j).
\]
Finally, since we know that the lower bound holds on open subsets of $\R \setminus [-j^*,j^*]$, for a generic $O \subset \R$
\[
\varliminf_{t\to+\infty} \frac1t \, \log \mu_t(O) \geq \max\left(-\inf_{j\in O \setminus [-j^*,j^*]} I(j), -\inf_{j\in O \cap [-j^*,j^*]} I(j)\right)=-
\inf_{j \in O} I(j).
\]

\smallskip{\it Proof of \eqref{ent} for $j\in\, [0,j^*]$}.%
Let us denote $\phi^+:=\phi_{\beta_L}$, $\phi^-:=\phi_{\beta_R}$.
Let us set $\alpha:=j/j^*\in[0,1]$ and for $\sigma\in\{+,-\}$
\[
\gamma^\sigma(dp):= \frac1{\phi^\sigma(1/p)}\, \frac1{p^\sigma}
\, \phi^\sigma(dp)
\]
and for $\epsilon>0$
\begin{equation}
\label{rho}
\rho(p)\equiv \rho^\epsilon(p):=(1-\alpha) \, \frac{\un{]0,\epsilon]}(p)}{\gamma^\sigma(]0,\epsilon])}
        +\alpha \,\frac{\un{]\epsilon,+\infty[} (p)}{\gamma^\sigma(]\epsilon,+\infty[)}, 
\end{equation}
and
\begin{equation}\label{pipi}
\pi^\sigma(dp):=\rho(p)\, \gamma^\sigma(dp), \qquad \tilde\pi^\sigma:= \frac1{\pi^\sigma(p)} \, p\,
\pi^\sigma(dp).
\end{equation}
Notice that $\pi^\sigma\rightharpoonup\alpha\gamma^\sigma+(1-\alpha)\delta_{0}$ while $\tilde\pi^\sigma\rightharpoonup\phi^\sigma$,  
as $\epsilon\downarrow 0$. In particular, the weak limit of $\pi^\sigma$ as $\epsilon\downarrow 0$ depends on $\alpha$, while
the limit of $\tilde\pi^\sigma$ does not. This property is crucial to the argument, in particular in \eqref{use} below.

We denote by $\bbP$ the law on $(\R_+^{\N^*})$ such that under
$\bbP$ the sequence $(\tau_i)_{i\geq 1}$ is independent and
\begin{enumerate}
\item for all $i\in2\N$, $1/\tau_i$ has law $\phi^-$
\item for all $i\in2\N+1$, $1/\tau_i$ has law $\phi^+$.
\end{enumerate}
and by $\bbP_{\tilde\pi}$ the law on $(\R_+^{\N^*})$ such that under
$\bbP_{\tilde\pi}$ the sequence $(\tau_i)_{i\geq 1}$ is independent and
\begin{enumerate}
\item for all $i\in2\N$, $1/\tau_i$ has law $\tilde\pi^-$
\item for all $i\in2\N+1$, $1/\tau_i$ has law $\tilde\pi^+$.
\end{enumerate}

For $\eta>0$ we also define $T_t:=\lfloor \frac{2(1+\eta)}{\tilde\pi^+(p)+\tilde\pi^-(p)}t\rfloor$; and let $\bbP^{t,\epsilon,\eta}$ be the law on 
$(\R_+^{\N^*})$ such that under $\bbP^{t,\epsilon,\eta}$ the sequence $(\tau_i)_{i\geq 1}$ is independent and
\begin{enumerate}
\item for all $i=1,\ldots, T_t$ and $i\in2\N$, $1/\tau_i$ has law $\tilde\pi^-$
\item for all $i=1,\ldots, T_t$ and $i\in2\N+1$, $1/\tau_i$ has law $\tilde\pi^+$
\item if $i\geq T_t+1$ and $i\in2\N$, then $1/\tau_{i}$ has distribution $\phi^-$
\item if $i\geq T_t+1$ and $i\in2\N+1$, then $1/\tau_{i}$ has distribution $\phi^+$.
\end{enumerate}
Let us denote by $\nu_{t,\epsilon,\eta}$ the law of $J[0,t]/t$
under $\bbP^{t,\epsilon,\eta}$. Let us prove now that
\begin{equation}\label{sieg}
\lim_{\epsilon\downarrow0}\lim_{t\uparrow +\infty} \nu_{t,\epsilon,\eta} = \delta_j.
\end{equation}
By the law of large numbers, under $\bbP_{\tilde\pi}$ we have a.s.
\[
\lim_{t\to+\infty} \frac{S_{T_t}}t =
\lim_{t\to+\infty} \frac{S_{T_t}}{T_t} \,\frac{T_t}t=
\frac{\pi^+(p)+\pi^-(p)}2\, \frac{2(1+\eta)}{\tilde\pi^+(p)+\tilde\pi^-(p)}=1+\eta>1.
\]
However $S_{T_t}$
has the same law under $\bbP_{\tilde\pi}$ and under $\bbP^{t,\epsilon,\eta}$, so we obtain
for any $\epsilon,\,\eta>0$
\[
\lim_{t\to+\infty} \bbP^{t,\epsilon,\eta}\left(S_{T_t}\leq t\right)= 
\lim_{t\to+\infty} \bbP_{\tilde\pi}\left(\frac{S_{T_t}}t\leq 1\right)=0.
\]
or equivalently
\begin{equation*}
\lim_{t\to+\infty} \bbP^{t,\epsilon,\eta}\left(D_{t,\eta} \right)=1.
\end{equation*}
where 
\[
D_{t,\eta}:=\left\{S_{T_t}>t\right\}.
\]
Recall $\{S_n>t\}=\{N_t+1\leq n\}$, so that on $D_{t,\eta}$ we have $N_t+1\leq T_t$ . Therefore for any $\epsilon,\,\eta>0$
\begin{equation}
\label{nmire}
\begin{split}
\varlimsup_{t \to +\infty} \bbP^{t,\epsilon,\eta}(|J[0,t]/t-j|>\gep) & \leq \varlimsup_{t \to +\infty} \bbP_{\tilde\pi}(\{|J[0,t]/t-j|>\gep\}\cap D_{t,
\eta}) +
\bbP^{t,\epsilon,\eta}(D_{t,\eta}^c) 
\\ &
= \varlimsup_{t \to +\infty} \bbP_{\tilde\pi}(\{|J[0,t]/t-j|>\gep\}\cap D_{t,\eta}) 
\end{split}
\end{equation}
By the law of the large numbers and by the renewal theorem \cite[Proposition V.1.4]{Asmussen}, we have $\bbP_{\tilde\pi}$ a.s. as $t
\uparrow+\infty$
\[
\begin{split}
\frac1t J[0,t] & =\frac1t \hf\sum_{k= 1}^{N_t}
v^2_{k} (-1)^{k+1} = \hf\frac{N_t}t \, \frac1{N_t}\sum_{k= 1}^{N_t}
v^2_{k} (-1)^{k+1}\to \hf\frac{\E_{\tilde\pi}(\tau_1^{-2}-\tau_2^{-2})}{\E_{\tilde\pi}(\tau_1+\tau_2)}
\\ & =\hf\frac{{\tilde\pi}^+(p^{2})-{\tilde\pi}^-(p^{2})}{{\tilde\pi}^+(1/p)+{\tilde\pi}^-(1/p)}=\hf\frac{{\tilde\pi}^+(p^{2})-{\tilde\pi}^-(p^{2})}{\frac1
{\pi^+(p)}+\frac1{\pi^-(p)}}
\end{split}
\]
and as $\epsilon\downarrow0$
\begin{equation}\label{use}
\hf\frac{{\tilde\pi}^+(p^{2})-{\tilde\pi}^-(p^{2})}{\frac1{\pi^+(p)}+\frac1{\pi^-(p)}}
 \to \, 
\hf\frac{{\phi}^+(p^{2})-{\phi}^-(p^{2})}{\frac1{\alpha\gamma^+(p)}+\frac1{\alpha\gamma^-(p)}}=\alpha\frac{\beta_L^{-1}-\beta_R^{-1}}
{\left(\frac{\pi\beta_L}2\right)^{\frac12}+\left(\frac{\pi\beta_R}2\right)^{\frac12}}=\alpha j^*=j,
\end{equation}
so that by \eqref{nmire}
\begin{equation}
\label{converppi}
\lim_{\epsilon\downarrow 0}\lim_{t\uparrow+\infty} \bbP_{\tilde\pi}\left(|J[0,t]/t-j|>\gep\right)=0
\end{equation}
which yields \eqref{sieg}. 

Now we estimate the entropy
\begin{equation}
\label{e:Hbound}
 \begin{split}
\varlimsup_{t\uparrow +\infty} \frac 1t \H(\nu_{t,\epsilon,\eta} \,| \, \mu_t)  & \le 
\varlimsup_{t\uparrow +\infty} \frac 1t \H\left(\bbP^{t,\epsilon,\eta} \, |\, \bbP\right)
\\ & = \varlimsup_{t\uparrow +\infty}
\frac 1t \sum_{i=1}^{T_t} \left(\un{(i\in2\N+1)}\, \H(\tilde\pi^+ \, | \, \phi^+)
+\un{(i\in2\N)}\, \H(\tilde\pi^- \, | \, \phi^-)\right),
\end{split}
\end{equation}
so that
\begin{equation}
\label{ent3}
\begin{split}
\varlimsup_{\epsilon \downarrow 0} \varlimsup_{t\uparrow +\infty}
\frac 1t \H(\nu_{t,\epsilon,\eta} \,| \, \mu_t) & \le \varlimsup_{\epsilon \downarrow 0}  \frac{2(1+\eta)}{\tilde\pi^+(p)+\tilde\pi^-(p)} \frac{\H
(\tilde\pi^+ \, | \, \phi^+) + \H(\tilde\pi^- \, | \, \phi^-)}2.
\end{split}
\end{equation}
Now, recalling that $\rho:=\frac{d\pi^+}{d\gamma^+}$ by \eqref{rho}-\eqref{pipi}, a tedious explicit computation based on the definition of $\pi^
\sigma$ shows that
\[
\H(\tilde\pi^+ \, | \, \phi^+)= \int \log\left(\frac{\rho}{\phi(\rho)}\right)\, d\tilde\pi^+
=-\log(\phi(\rho))+\int\log \rho \, d\tilde\pi^+ \to -\log\alpha+\log\alpha=0
\]
as $\epsilon \downarrow 0$, and analogously for $\H(\tilde\pi^- \, | \, \phi^-)$. Moreover, arguing as above
\begin{equation}\label{eta}
\varlimsup_{\eta \downarrow 0} \varlimsup_{\epsilon \downarrow 0}\frac{2(1+\eta)}{\tilde\pi^+(p)+\tilde\pi^-(p)} = 2\alpha\kappa <+\infty.
\end{equation}
Thus
$$\lim_{\epsilon \downarrow 0} \varlimsup_{t\uparrow +\infty}
\frac 1t \H(\nu_{t,\epsilon,\eta} \,| \, \mu_t)=0 $$
which implies, together with \eqref{converppi}, that there exists a map $t\mapsto (\epsilon(t),\eta(t))$ vanishing as $t \uparrow +\infty$ such that
$\nu_t:= \nu_{t,\epsilon(t),\eta(t)} \to \delta_j$ and $\varlimsup_t t^{-1}
\H(\nu_t \,| \, \mu_t) = 0$.

\smallskip{\it Proof of \eqref{ent} for
$j\in\, [-j^*,0]$}. Now the strategy is very similar, but we {\it reverse} the current constructed for $j\in\, [0,j^*]$. This mirrors the Gallavotti-
Cohen symmetry relation \eqref{GCI}.

Set now $\alpha:=-j/j^*\in[0,1]$, $\pi^\sigma$ as in \eqref{rho}-\eqref{pipi} and 
$T_t:=\lfloor \frac{2(1+\eta)}{\tilde\pi^+(p)+\tilde\pi^-(p)}t\rfloor$ with $\eta>0$. 
Let us define by $\bbP^{t,\epsilon,\eta}$ the law on $(\R_+^{\N^*})$ such that under
$\bbP^{t,\epsilon,\eta}$ the sequence $(\tau_i)_{i\geq 1}$ is independent and
\begin{enumerate}
\item for all $i=1,\ldots, T_t$ and $i\in2\N$, $v_i$ has law $\tilde\pi^+$
\item for all $i=1,\ldots, T_t$ and $i\in2\N+1$, $v_i$ has law $\tilde\pi^-$
\item if $i\geq T_t+1$ and $i\in2\N$, then $v_{i}$ has distribution $\phi^-$
\item if $i\geq T_t+1$ and $i\in2\N+1$, then $v_{i}$ has distribution $\phi^+$.
\end{enumerate}
Let us denote by $\nu_{t,\epsilon,\eta}$ the law of $J[0,t]/t$
under $\bbP^{t,\epsilon,\eta}$. Arguing as in the proof of \eqref{sieg}, we obtain that
under $\bbP_{\tilde\pi}$,
a.s. as $t\uparrow+\infty$
\[
\frac1t J[0,t] \to 
\hf\frac{{\tilde\pi}^-(p^{2})-{\tilde\pi}^+(p^{2})}{{\tilde\pi}^+(1/p)+{\tilde\pi}^-(1/p)}=\hf\frac{{\tilde\pi}^+(p^{2})-{\tilde\pi}^-(p^{2})}{\frac1{\pi^+
(p)}+\frac1{\pi^-(p)}}
\]
and as $\epsilon\downarrow0$
\[
\hf\frac{{\tilde\pi}^-(p^{2})-{\tilde\pi}^+(p^{2})}{\frac1{\pi^+(p)}+\frac1{\pi^-(p)}}
 \to \, -\alpha j^*=j,
\]
so that indeed  for any $\varepsilon>0$
\[
\lim_{\epsilon\downarrow 0}\lim_{t\uparrow+\infty} \bbP^{t,\epsilon,\eta}\left(|J[0,t]/t-j|>\gep\right)=0
\]
and therefore
\[
\lim_{\epsilon\downarrow 0}\lim_{t\uparrow +\infty} \nu_{t,\epsilon,\eta} = \delta_j.
\]
Now we estimate the entropy, arguing as in \eqref{e:Hbound},
\[
\begin{split}
\lim_{\epsilon \downarrow 0} \varlimsup_{t\uparrow +\infty}
\frac 1t \H(\nu_{t,\epsilon,\eta} \,| \, \mu_t) & \le \varlimsup_{\epsilon \downarrow 0}  \frac{2(1+\eta)}{\tilde\pi^+(p)+\tilde\pi^-(p)}\,\frac{ \H
(\tilde\pi^- \, | \, \phi^+) + \H(\tilde\pi^+ \, | \, \phi^-)}2.
\end{split}
\]
Now, recalling the definition of $\rho$ in \eqref{rho}-\eqref{pipi}, we have
\[
\frac{d\tilde\pi^-}{d\phi^+}=\frac{d\tilde\pi^-}{d\phi^-}\,\frac{d\phi^-}{d\phi^+}
= \frac{\rho}{\phi^-(\rho)}\, \frac{\beta_R}{\beta_L}\, e^{-(\beta_R-\beta_L)\, 
\frac{p^2}2}
\]
so that
\[
\H(\tilde\pi^- \, | \, \phi^+) = -\log\phi^-(\rho)+ \int \log \rho \, d\tilde\pi^-
+\log\left(\frac{\beta_R}{\beta_L}\right)-(\beta_R-\beta_L)\int 
\frac{p^2}2\, \tilde\pi^-(dp).
\]
As before, when $\epsilon \downarrow 0$ we have 
$-\log\phi^-(\rho)+ \int \log \rho \, d\tilde\pi^-\to-\log\alpha+\log\alpha=0$.
Now, as $\epsilon \downarrow 0$
\[
\int \frac{p^2}2\, \tilde\pi^-(dp)\to \int \frac{p^2}2\, \phi^-(dp) = \frac1{\beta_R},
\]
so that
\[
\lim_{\epsilon \downarrow 0}\H(\tilde\pi^- \, | \, \phi^+) =
\H(\phi^- \, | \, \phi^+) = \log\left(\frac{\beta_R}{\beta_L}\right)
-\frac{\beta_R-\beta_L}{\beta_R}=\log\left(\frac{\beta_R}{\beta_L}\right)
-(\beta_R-\beta_L)T_R,
\]
and analogously
\[
\lim_{\epsilon \downarrow 0}\H(\tilde\pi^+ \, | \, \phi^-) =
\log\left(\frac{\beta_L}{\beta_R}\right)
-(\beta_L-\beta_R)T_L.
\]
Therefore, by \eqref{eta},
\[
\begin{split}
\varlimsup_{\eta \downarrow 0} \varlimsup_{\epsilon \downarrow 0} \varlimsup_{t\uparrow +\infty}
\frac 1t \H(\nu_{t,\epsilon,\eta} \,| \, \mu_t) & \le  2\alpha\kappa\,\frac{(\beta_R-\beta_L)(T_L-T_R)}2=(\beta_R-\beta_L)|j|=I(j).
\end{split}
\]
Then again, there exists a map $t\mapsto (\epsilon(t),\eta(t))$ vanishing as $t \uparrow +\infty$ such that
$\nu_t:= \nu_{t,\epsilon(t),\eta(t)} \to \delta_j$ and $\varlimsup_t t^{-1}
\H(\nu_t \,| \, \mu_t)\leq  I(j)$.
\end{proof}

\section{Scaling limit}
It is natural to consider the difference of temperatures as a variable and we introduce:
$$
\cF(\lambda,\tau,T):=f\left(\lambda,\frac{1}{T+\frac{\tau}{2}},\frac{1}{T-\frac{\tau}{2}}\right),
$$
recall \eqref{freeener}, 
and the corresponding Legendre transform:
\[
\cI(j,\tau,T)=\sup_{\lambda}\{\lambda j-\cF(\lambda,\tau,T)\}.
\]
Then we have
$$
\varepsilon^{-2}\cI(\varepsilon j,\varepsilon\tau,T)=\sup_{\lambda}\{\lambda j-\varepsilon^{-2}\cF(\varepsilon\lambda,\varepsilon\tau,T)\}.
$$
The central result necessary to study the scaling limit is the following
\begin{prop} \label{scaling}  Let $\kappa=(\frac{T}{2\pi})^\hf$.  
\begin{enumerate}
\item If $\tau\ne0$,
then
\begin{equation}\label{fepsilon}
\lim_{\varepsilon\downarrow 0}\varepsilon^{-2}\cF(\varepsilon\lambda,\varepsilon\tau,T)=\cH(\lambda,\tau,T):=
\left\{
\begin{array}{ll}
\lambda \kappa\tau+\kappa\lambda^2 T^2 \ \ {\rm if} \ \ \lambda\tau>0,
 \\
0  \ \ {\rm if} \ \
\lambda\tau\in[-\tau^2,0],
\\
-(\lambda+\tau) \kappa\tau+\kappa(\tau+\lambda)^2 T^2 \ \ {\rm if} \ \ \lambda\tau<-\tau^2
\end{array}
\right.
\end{equation}
%
%
and
\begin{equation}\label{iepsilon}
\lim_{\varepsilon\downarrow 0}\varepsilon^{-2}\cI(\varepsilon j,\varepsilon\tau,T)=\cG(j,\tau,T):=\left\{\begin{array}{ll}
\frac{(j-\kappa\tau)^2}{4\kappa T^2} \ \ {\rm if} \ \ j\tau>\kappa\tau^2  \\
0  \ \ {\rm if} \ \ j\tau\in[0,\kappa\tau^2] 
\\ 
\frac{-j\tau}{2 T^2} \ \ {\rm if} \ \ j\tau\in[-\kappa\tau^2,0] 
\\
\\
\frac{j^2+\kappa^2\tau^2}{4\kappa T^2} \ \ {\rm if} \ \ j\tau<-\kappa\tau^2.
\end{array}
\right.
\end{equation}
\item If $\tau=0$, $\lim_{\varepsilon\downarrow 0}\varepsilon^{-2}\cF(\varepsilon\lambda,\varepsilon\tau,T)=\kappa\lambda^2T^2$ and $\lim_{\varepsilon\downarrow 0}\varepsilon^{-2}\cI(\varepsilon j,\varepsilon\tau,T)=\frac{j^2}{4\kappa T^2}$.
\end{enumerate}

The above convergences are uniform in $(\lambda,\tau,T)$, resp. $(j,\tau,T)$. 
\end{prop}
The plot of $\cG$ was given in figure 1 and we give a plot of $\cH$ in Figure 6.
\begin{figure}[thb]
\includegraphics[width =.50\textwidth]{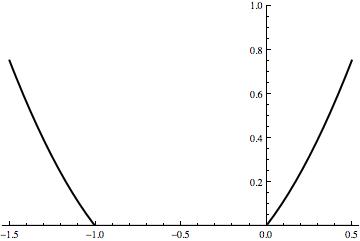} 
\caption*{ Figure 6: Plot of $\cH$ as a function of $\lambda$ for $\kappa=\tau= T^2=1$}
\end{figure}

\begin{proof}
We assume that $\tau>0$, the case $\tau<0$ being
completely analogous. We show first that $\forall \lambda>0$, 
$$
\lim_{\varepsilon\downarrow 0}\varepsilon^{-2}\cF(\varepsilon\lambda,\varepsilon\tau,T)=-\lambda \kappa\tau+\kappa\lambda^2 T^2.
$$
We set $g(\varepsilon):=\cF(\varepsilon\lambda,\varepsilon\tau,T)$ and use the notation:
$$
g(0^+)=\lim_{t\downarrow 0} g(t),\quad g(0^-)=\lim_{t\uparrow 0} g(t).
$$
By proposition (\ref{mgf}) and the implicit function theorem $\cF(\varepsilon\lambda,\varepsilon\tau,T)$ is a smooth function of $\varepsilon$ for $\varepsilon\geq 0$. Thus, we write for $\varepsilon> 0$ small enough,
\begin{equation}
g(\varepsilon)=g(0^+)+\varepsilon g'(0^+)+\frac{\varepsilon^2}{2} g''(0^+)+O(\varepsilon^3),
\label{expand}
\end{equation}
uniformly in $(\lambda,\tau,T)$.
Since $\cF(\cdot,0,T)$ is continuous and $\cF(0,0,T)=0$ by proposition \ref{mgf}, we get   $\lim_{\varepsilon\downarrow 0}\cF(\varepsilon\lambda,\varepsilon\tau,T)=0$, i.e. $g(0^+)=0$.

 We compute now the derivatives $g'(0^+)$ and $g''(0^+)$ using expression (\ref{F}) and the relation (\ref{F1}).   In order to simplify notations we introduce the distribution over $\R_+^2$:
\begin{equation}
\psi(v_1,v_2)=\beta_L\beta_R\, v_1\, v_2\,
\exp\left(-\frac{\eta}{v_1} -\frac{\eta}{v_2}-\left(\beta_L-\lambda\right)\, \frac{v_1^2}2
-\left(\beta_{R}+\lambda\right)\, \frac{v_2^2}2\right) 
\end{equation}
where $\eta$ is chosen such that the distribution is normalized, i.e (\ref{F}) holds.  We denote by $\E_{\beta_L,\beta_R,\lambda}(h)$ expectation of a Borel function $h$ with respect to $\psi$.
Below $\eta$ is chosen so that (\ref{F}) holds with the rescaled variables, namely $\eta=\eta^\varepsilon=\cF(\varepsilon\lambda,\varepsilon\tau,T)$. 

\vspace{3mm}

\noindent Let us start with the first derivative:
\begin{equation}
g'(0^+)=\lim_{\varepsilon\downarrow 0}\frac{d}{d\varepsilon}\cF(\varepsilon\lambda,\varepsilon\tau,T)=\lim_{\varepsilon\downarrow 0}\left (\frac{\partial \cF}{\partial\lambda}(\varepsilon\lambda,\varepsilon\tau,T)\lambda+\frac{\partial \cF}{\partial\tau}(\varepsilon\lambda,\varepsilon\tau,T)\tau\right ).
\label{gp0}
\end{equation}
Taking first the derivative with respect to $\lambda$, we get from (\ref{F}), 
$$
\frac{\partial F}{\partial\eta}\frac{\partial \cF}{\partial\lambda}+\frac{\partial F}{\partial\lambda}=0.
$$
Next one  computes explicitly,
\begin{equation}
 \frac{\partial F}{\partial\lambda}(\varepsilon\lambda,\eta,\beta_L(\varepsilon\tau),\beta_R(\varepsilon\tau))=\hf \E_{\beta_L(\varepsilon\tau),\beta_R(\varepsilon\tau),\varepsilon\lambda}( v_2^2-v_1^2).
\label{partiallam}
\end{equation}
with $\beta_L(\tau)=(T+\tau)^{-1}$ and $\beta_R(\tau)=(T-\tau)^{-1}$.  Since $\E_{\beta,\beta,0}( v_2^2-v_1^2)=0$, we get 
$$
\lim_{\varepsilon\downarrow 0}\frac{\partial F}{\partial\lambda}(\varepsilon\lambda,\eta,\beta_L(\varepsilon\tau),\beta_R(\varepsilon\tau))=0
$$
by dominated convergence.
 Taking next the derivative with respect to $\tau$, we obtain from (\ref{F}),
\begin{equation}
\frac{\partial F}{\partial\eta}\frac{\partial \cF}{\partial\tau}+\frac{\partial F}{\partial\beta_L}\frac{\partial\beta_L}{\partial \tau}+\frac{\partial F}{\partial\beta_R}\frac{\partial\beta_R}{\partial \tau}=0,
\label{derif}
\end{equation}
and
\begin{equation}
\lim_{\varepsilon\downarrow 0}\frac{\partial F}{\partial\beta_L}(\varepsilon\lambda,\eta,\beta_L(\varepsilon\tau),\beta_R(\varepsilon\tau))=\left(T-\hf \E_{\beta,\beta,0}(v_1^2)\right )=0.
\label{derif1}
\end{equation}
Similarly,
\begin{equation}
\lim_{\varepsilon\downarrow 0}\frac{\partial F}{\partial\beta_R}(\varepsilon\lambda,\eta, \beta_L(\varepsilon\tau),\beta_R(\varepsilon\tau))=\left(T-\hf \E_{\beta,\beta,0}(v_2^2)\right )=0.
\label{derif2}
\end{equation}
Using those two expressions in (\ref{derif}), we get
$$
\lim_{\varepsilon\downarrow 0}\frac{\partial}{\partial\tau}\cF(\varepsilon\lambda,\varepsilon\tau,T)=0.
$$
Combining this with (\ref{partiallam}) in (\ref{gp0}), we finally obtain that $g'(0^+)=0$.

\vspace{3mm}

\noindent We next compute the second derivative with respect to $\varepsilon$
\begin{eqnarray}
g''(0^+)&=&\lim_{\varepsilon\downarrow 0}\frac{d^2}{d\varepsilon^2}\cF(\varepsilon\lambda,\varepsilon\tau,T)\nonumber\\
&=&\lim_{\varepsilon\downarrow 0}\left (\frac{\partial^2\cF}{\partial \lambda^2}(\varepsilon\lambda,\varepsilon\tau,T)\lambda^2+\frac{\partial^2\cF}{\partial \lambda\partial\tau}(\varepsilon\lambda,\varepsilon\tau,T)\lambda\tau+\frac{\partial^2\cF}{\partial \tau^2}(\varepsilon\lambda,\varepsilon\tau,T)\tau^2\right).
\label{second}
\end{eqnarray}
 The first two derivatives in (\ref{second}) have been computed in \cite{LefevereZambotti1}, the result is:
\begin{equation}
 \lim_{\varepsilon\downarrow 0}\left (\frac{\partial^2\cF}{\partial \lambda^2}(\varepsilon\lambda,\varepsilon\tau,T)\lambda^2+\frac{\partial^2\cF}{\partial \lambda\partial\tau}(\varepsilon\lambda,\varepsilon\tau,T)\lambda\tau\right )=2(\kappa T^2\lambda^2+\kappa \lambda\tau).
 \label{rappel}
\end{equation} 
We show now that the second derivative with respect to $\tau$ vanishes when $\varepsilon$ goes to zero. 
In (\ref{F1}), we write $\beta_L$ and $\beta_R$ as function of $\tau$ and take derivatives with respect to $\tau$.
\begin{eqnarray}
\frac{\partial ^2\cF}{\partial \tau^2}=
&-&\frac{\partial ^2 F}{\partial \eta^2}\left (\frac{\partial  F}{\partial \beta_R}\beta_R^2-\frac{\partial  F}{\partial \beta_L}\beta_L^2\right )
\left(\frac{\partial  F}{\partial \eta}\right)^{-2}\nonumber\\
&-&2\left (\frac{\partial  F}{\partial \beta_R}\beta_R^3+\frac{\partial  F}{\partial \beta_L}\beta_L^3\right )\left(\frac{\partial  F}{\partial \eta}\right)^{-1}
\nonumber\\
&+&\left(2\frac{\partial^2 F}{\partial \beta_L\partial \beta_R}\beta_L\beta_R-\frac{\partial^2 F}{\partial \beta^2_L}\beta_L^2-\frac{\partial^2 F}{\partial \beta^2_R}\beta_R^2\right) \left(\frac{\partial  F}{\partial \eta}\right)^{-1}.
\end{eqnarray}
By (\ref{derif1}) and (\ref{derif2})  the first derivatives of $F$ with respect to $\beta_L$ and $\beta_R$ vanishes when $\varepsilon\downarrow 0$.  We compute now the second derivatives with respect to $\beta_L$ and $\beta_R$ and show that each also vanishes when $\varepsilon\downarrow 0$, for instance,
$$
\lim_{\varepsilon\downarrow 0}\frac{\partial^2 F}{\partial \beta^2_L}(\varepsilon\lambda,\eta,\beta_L(\varepsilon\tau),\beta_R(\varepsilon\tau))=\left (-T^2+\frac{1}{4}\E_{\beta,\beta,0} (v_1^4)-\hf T \,\E_{\beta,\beta,0} (v_1^2)\right )=0.
$$
All other derivatives may be dealt with in the same way and thus
$$
\lim_{\varepsilon\downarrow 0}\frac{\partial ^2\cF}{\partial \tau^2}(\varepsilon\lambda,\varepsilon\tau,T)=0.
$$
Combining this with (\ref{rappel}) in (\ref{second}), we finally get for the second derivative,
$$
g''(0^+)=2(\kappa T^2\lambda^2-\kappa \lambda\tau).
$$
Plugging this last expression in (\ref{expand}), we finally obtain the result (\ref{fepsilon}) for $\lambda>0$.
\vspace{3mm}

\noindent Let us now consider the case $\lambda<-\tau<0$. By the Gallavotti-Cohen
symmetry relation \eqref{GC} and the definition of $\cF(\lambda,\tau,T)$, we obtain that for $\lambda<-\tau$,
\[
\varepsilon^{-2}\cF(\varepsilon\lambda,\varepsilon\tau,T) = 
\varepsilon^{-2}\cF(\varepsilon(\tau-\lambda),\varepsilon\tau,T),
\]
and therefore by (\ref{fepsilon}) valid for $\lambda>0$, we get:
\[
\lim_{\varepsilon\downarrow 0}\varepsilon^{-2}\cF(\varepsilon\lambda,\varepsilon\tau,T)=
(\lambda-\tau) \kappa\tau+\kappa(\tau-\lambda)^2 T^2 \ \ {\rm if} \ \ \lambda<-\tau.
\]
Finally, if $\lambda\in[-\tau,0]$, then $\cF(\varepsilon\lambda,\varepsilon\tau,T)\equiv 0$ by
point (2) of Proposition \ref{mgf}. Formula \eqref{fepsilon} is now completely proved. The case $\tau=0$ is exactly analogous.

\medskip
We want to prove now \eqref{iepsilon}.  We assume again that $\tau>0$ and  we first note that for any $\varepsilon>0$,
\begin{equation}
\varepsilon^{-2}\cI(\varepsilon j,\varepsilon\tau,T)
=\sup_{\lambda}\{\lambda j-\varepsilon^{-2}\cF(\varepsilon j,\varepsilon\tau,T)\}=\left\{\begin{array}{ll}
0  \ \ {\rm if} \ \ j\in[0,\kappa\tau] 
\\ 
\frac{-j\tau}{2 T^2} \ \ {\rm if} \ \ j\in[-\kappa\tau,0] 
\\
\end{array}
\right.
\end{equation}
This follows from Lemma \ref{GCI0}.

\noindent When $j>\kappa\tau$, one notices that the maximizer $\lambda(\varepsilon,j)$ of
$
\varepsilon^2\lambda j-\cF(\varepsilon\lambda,\varepsilon\tau,T)
$
is such that $0<\lambda(\varepsilon,j)<+\infty$ and is the solution of the implicit equation in the unknown $\lambda$:
\begin{equation}
\varepsilon j=\frac{\partial\cF}{\partial \lambda}(\varepsilon\lambda,\varepsilon\tau,T).
\label{maxeq}
\end{equation}
By the implicit function theorem, $\lambda(\varepsilon,j)$ is therefore a smooth  function of $\varepsilon\geq 0$.  By performing computations similar to the above ones, namely expanding (\ref{maxeq}) in $\varepsilon$  one can show that $\lambda(0^+,j)=\frac{j-\kappa\tau}{2\kappa T^2}$ which is the maximizer of the expression
$\lambda j-\cH(\lambda,\tau,T)$, where $\cH$ has been defined in (\ref{fepsilon}).
For each $\varepsilon>0$, $\varepsilon^{-2}\cF(\varepsilon\lambda,\varepsilon\tau,T)$ is a convex function (as a function of $\lambda$), thus the convergence to $\cH(\lambda,\tau,T)$ is uniform and
$$
\lim_{\varepsilon\downarrow 0}\varepsilon^{-2}\cF(\varepsilon\lambda(\varepsilon,j),\varepsilon\tau,T)=\cH(\lambda(0^+,j),\tau,T).
$$
Therefore,
\begin{equation}
\lim_{\varepsilon\downarrow 0}\varepsilon^{-2}\cI(\varepsilon j,\varepsilon\tau,T)=\lambda(0^+,j) j-\cH(\lambda(0^+,j),\tau,T)=\cG(j,\tau,T).
\end{equation}
The case $j<-\kappa\tau$ is obtained from the case $j>\kappa\tau$ by using the Gallavotti-Cohen symmetry of Lemma \ref{GCI0}.
\end{proof}

\section{Appendix: The generator}

To avoid too heavy notations, we consider the simpler case of a particle moving in the interval $[0,1[$ with a positive velocity.  When the particle reaches $1$, it is absorbed and re-emitted in $0$ with a random positive velocity distributed with a density $\varphi$.  This dynamics has been introduced in section 2 of \cite{LefevereZambotti1} and we follow the notations introduced there. The statement of proposition \ref{infinitesimal2} is a trivial adaptation of proposition \ref{infinitesimal} below.
We want to compute the infinitesimal generator, or more precisely the
infinitesimal action of the dynamics on a smooth function (which is not
necessarily in the domain of the generator).
In other words, our aim is to prove the following result, where
we denote
the law of $\tau_i$ by $\psi(d\tau)$ and the law of $v_i=1/\tau_i$
by $\varphi(du)$.  
\begin{prop}\label{infinitesimal}
For all $f,g:[0,1]\times\R_+\mapsto\R$ bounded with bounded
continuous first derivatives:
\[
\begin{split}
& \left. \frac d{dt} \, \int_0^1 dq\, \int_{\R_+}dp \,
g(q,p)\, P_tf(q,p) \right|_{t=0} =
\\ & = \int_{\R_+} dp \, \int_0^1 dq \, g(q,p)\, p \,
f_q(q,p) + \int_{\R_+} dp \, p\, g(1,p) \int_{\R_+}
\varphi(du)\, (f(0,u)-f(1,p))
\end{split}
\]
\end{prop}
\begin{proof} The law of $\tau_1+\cdots+\tau_n$ is denoted
as usual by the $n$-fold convolution $\psi^{n*}$ and we recall that
$S_n=S_0+\tau_1+\cdots+\tau_n$.
Then we can write
\[
\begin{split}
& P_tf(q_0,p_0)=
\\ & = \bo{(t<S_0 )} \, f(q_0+p_0t,p_0)+ \bo{(t\geq S_0)} \, \sum_{n=1}^\infty
\E\left(\bo{(S_{n-1}\leq t< S_n)} \,
f\left(\frac{t-S_{n-1}}{\tau_n}, \frac1{\tau_n}\right)\right)
\\ & = \bo{(t< S_0)} \, f(q_0+p_0t,p_0)+ \\
& \qquad \qquad + \bo{(t\geq S_0)} \, \sum_{n=1}^\infty
\int_{[0,t-S_0]} \psi^{*(n-1)}(ds) \int_{]t-S_0-s,+\infty[} \psi(d\tau)\,
f\left(\frac{t-S_0-s}{\tau}, \frac1{\tau}\right)
\\ & = \bo{(t< S_0)} f(q_0+p_0t,p_0)+ \bo{(t\geq S_0)}
\int_{[0,t-S_0]} U(ds) \int_{]t-S_0-s,+\infty[} \psi(d\tau)\,
f\left(\frac{t-S_0-s}{\tau}, \frac1{\tau}\right)
\end{split}
\]
where we set
$\psi^{*0}(ds)=\delta_0(ds)$ and
\[
U[a,b] =\sum_{n=1}^\infty \int_a^b \psi^{*(n-1)}(ds)
= \delta_0[a,b] + \sum_{n=1}^\infty \int_a^b \psi^{*n}(ds), \quad
0\leq a\leq b.
\]
The {\it renewal measure} $U(ds)$ gives the average number of
collisions in the time interval $ds$. 
We define accordingly
\[
\begin{split}
I_1(t) &:= \int_{[0,1[\times\R_+} dp\, dq\, g(q,p)
\, \bo{(t<S_0(q,p))} \, P_tf(q,p) \\ & 
= \int_{\R_+} dp \int_0^{1-tp} dq \, g(q,p) \, f(q+tp,p),
\end{split}
\]
\[
\begin{split}
I_2(t) &:=\int_{[0,1[\times\R_+} dp\, dq\, g(q,p)
\, \bo{(t\geq S_0(q,p))}\, P_tf(q,p)
\\ & = \int_{[0,1[\times\R_+} dp\, dq\, g(q,p) \, \bo{(q\geq 1-tp)}
 \int_0^{t-\frac{1-q}p} \, U(ds)
\int_s^{+\infty} \psi(dv)\,
f\left(\frac{t-\frac{1-q}p-s}{v}, \frac1{v}\right)
\\ & = \int_0^{+\infty} \psi(dv) \int_0^v U(ds) \int_{\R_+} dp
\int_{1\wedge(1-tp+sp)}^1 dq \, g(q,p) \,
f\left(\frac{t-\frac{1-q}p-s}{v}, \frac1{v}\right).
\end{split}
\]
Let us take the derivative in $t$
\[
\dot I_1(t)=\frac d{dt} \, I_1(t) = \int_{\R_+} dp \,
p\,\left[\int_0^{1-tp} dq \, g(q,p)\, f_q(q+tp,p)-
g(1-tp,p)\, f(1,p) \right],
\]
\[
\begin{split}
& \dot I_2(t)=\frac d{dt} \, I_2(t) =
\\ & = \int_0^{+\infty} \psi(dv) \int_0^v U(ds)
\int_{\R_+} dp \int_{1\wedge(1-tp+sp)}^1  dq \,
\frac 1v\, g(q,p) \,f_q\left(\frac{t-\frac{1-q}p-s}{v}, \frac1{v}\right)
\\ & + \int_0^{+\infty} \psi(dv) \int_0^v U(ds)
\int_{\R_+} dp \, p\, \bo{(1-tp+sp\leq 1)} \, g(1-tp+sp,p) \,
f\left(0, \frac1{v}\right).
\end{split}
\]
Let us let $t\to 0+$:
\[
\dot I_1(0) = \int_{\R_+} dp \, p\,\left[\int_0^1 dq
\, g(q,p)\, f_q(q,p)- g(1,p)\, f(1,p) \right]
\]
and since $U(ds)=\delta_0(ds)+\bo{]0,+\infty[}(s)\, U(ds)$

\[
\dot I_2(0) = \int_{\R_+} dp \, p\, g(1,p)
\int_{\R_+} \psi(dv)\,  f(0,v^{-1})= \int_{\R_+} dp \, p\, g(1,p)
\int_{\R_+} \varphi(du)\,  f(0,u).
\]
Therefore
\[
\begin{split}
& \left. \frac d{dt} \, \int_0^1 dq\, \int_{\R_+}dp \,
g(q,p)\, P_tf(q,p) \right|_{t=0} =
\\ & = \int_{\R_+} dp \, \int_0^1 dq \, g(q,p)\, p \,
f_q(q,p) + \int_{\R_+} dp \, p\, g(1,p) \int_{\R_+}
\varphi(du)\, (f(0,u)-f(1,p))
\end{split}
\]
\end{proof}
\noindent The infinitesimal action of the dynamics on a smooth function $f:[0,1]\times\R_+
\mapsto\R$ can therefore be written
\[
Lf(q,p) = p\, \frac{\partial f}{\partial q} + p\, \delta_1(dq) \,
\int_{\R_+} \varphi(du)\, (f(0,u)-f(1,p)).
\]

\subsection{The formal adjoint}
Let us suppose now that $\psi(dv)=\psi(v)\, dv$, so that the law of $v_i=1/\tau_i$ is
\[
\varphi(du)=\varphi(u)\, du=\frac{\psi(u^{-1})}{u^2}\,du.
\]
Then we can rewrite the result of Proposition \ref{infinitesimal} as follows
\[
\begin{split}
& \left. \frac d{dt} \, \int_0^1 dq\, \int_{\R_+}dp \,
g(q,p)\, P_tf(q,p) \right|_{t=0} 
\\ & =  \int_{\R_+} dp \left[- \int_0^1 dq \, g_q(q,p)\, p \, f(q,p) + p(g(1,p)f(1,p)-
g(0,p)f(0,p)) \right]
\\ & \quad  + \int_{\R_+} dp \, \varphi(p)\,
f(0,p) \int _{\R_+} du \, u\, g(1,u) - \int _{\R_+} dp \, p\, f(1,p)\, g(1,p)
\\ & = -\int_{\R_+} dp \int_0^1 dq \, g_q(q,p)\, p \, f(q,p) 
\\ & \quad - \int_{\R_+} dp \, f(0,p)\left[p\, g(0,p)-\varphi(p)\int _{\R_+} du \, u\, g(1,u)\right]
\end{split}
\]
and obtain an expression for the formal adjoint of $L$
\[
L^*g(q,p) = -p\, \frac{\partial g}{\partial q} - \delta_0(dq) \left[
p\, g(0,p) - \varphi(p) \int_{\R_+} du\, u\, g(1,u) \right].
\]
A solution of the Fokker-Planck equation associated with the process
$(q_t,p_t)$ must then satisfy the boundary condition
\[
p\, g(0,p) = \varphi(p) \int_{\R_+} du\, u\, g(1,u), \qquad \forall \ p>0.
\]
We can check that $g(q,p)= \bo{[0,1]}(q)\, \varphi(p)/(\mu p)$ is a probability
density solving the equation
$L^*g=0$. On the other hand, an invariant measure must satisfy $L^*g=0$.
Since $L^*g$ is a sum of two mutually singular measures, they must both vanish.
Then $g(q,p)=g(p)$ is constant in $q$ and
\[
p \, g(p) = \varphi(p) \int_{\R_+} du\, u\, g(u),
\]
i.e.
\[
g(p)= \frac1Z\, \frac{\varphi(p)}p.
\]

\end{document}